\documentclass[11pt]{article}
\usepackage{amsmath,amsthm,amssymb,mathtools,xspace,bm,multirow}
\usepackage{fullpage}

\usepackage[utf8]{inputenc} % allow utf-8 input
\usepackage[T1]{fontenc}    % use 8-bit T1 fonts
\usepackage[colorlinks=true]{hyperref}       % hyperlinks
\usepackage{url}            % simple URL typesetting
\usepackage{booktabs}       % professional-quality tables
\usepackage{amsfonts}       % blackboard math symbols
\usepackage{nicefrac}       % compact symbols for 1/2, etc.
\usepackage{microtype}      % microtypography

\usepackage{algorithm}
\usepackage{algorithmicx}
\usepackage[noend]{algpseudocode}
\usepackage[usenames,dvipsnames,table]{xcolor}
\usepackage[colorinlistoftodos,textsize=scriptsize]{todonotes}

\newtheorem*{theorem*}{Theorem}
\newtheorem{theorem}{Theorem}[section]
\newtheorem{lemma}[theorem]{Lemma}
\newtheorem{fact}[theorem]{Fact}
\newtheorem{definition}[theorem]{Definition}

\newtheorem{claim}[theorem]{Claim}
\newtheorem{corollary}[theorem]{Corollary}

\def\FullBox{\hbox{\vrule width 8pt height 8pt depth 0pt}} %I like full boxes.
\newcommand{\QED}{\;\;\;\FullBox}
\renewenvironment{proof}{\noindent{\bf Proof:~~}}{\hfill\QED}
\newenvironment{proofof}[1]{\noindent{\bf Proof of {#1}:~~}}{\(\QED\)}

\newcommand{\bmb}{{\bm b}}
\newcommand{\bmd}{{\bm d}}
\newcommand{\bmu}{{\bm u}}
\newcommand{\bmv}{{\bm v}}
\newcommand{\bmx}{{\bm x}}
\newcommand{\bmone}{{\bm 1}}
\newcommand{\eps}{\ensuremath{\epsilon}\xspace}
\newcommand{\eqdef}{\stackrel{\rm def}{=}}
\newcommand{\R}{\ensuremath{\mathbb{R}}\xspace}
\newcommand{\exxp}[1]{\underset{#1}{\mathbf{E}}}
\newcommand{\var}{\operatorname{\mathbf{Var}}}
\newcommand{\set}[1]{\{#1\}}
\newcommand{\norm}[1]{\lVert#1{\rVert}}
\newcommand{\normtwo}[1]{{\norm{#1}}_2}
\newcommand{\normF}[1]{{\norm{#1}}_F}
\newcommand{\Lovasz}{Lov{\'a}sz\xspace}

\newcommand{\dotprod}[2]{ \left\langle #1,\xspace #2 \right\rangle }      % <a,b>
\providecommand{\poly}{\operatorname*{poly}}
\providecommand{\diag}{\operatorname*{diag}}

\def\authornameal{Amit Levi}
\def\authoraffial{University of Waterloo. Email: \href{mailto: amit.levi@uwaterloo.ca}{amit.levi@uwaterloo.ca} . Research supported by NSERC Discovery grant and the David R. Cheriton Graduate Scholarship. Part of this work was done while the author was visiting NII Tokyo.}
\def\authornameyy{Yuichi Yoshida}
\def\authoraffiyy{National Institute of Informatics. Email: \href{mailto: yyoshida@nii.ac.jp}{yyoshida@nii.ac.jp}. Research supported by  JSPS KAKENHI Grant Number JP17H04676 and JST ERATO Grant Number JPMJER1201.}

\title{Sublinear-Time Quadratic Minimization via Spectral Decomposition of Matrices}

\author{
	\authornameal\thanks{\authoraffial} \and \authornameyy\thanks{\authoraffiyy}
}

\begin{document}

\maketitle
\begin{abstract}
We design a sublinear-time approximation algorithm for quadratic function minimization problems with a better error bound than the previous algorithm by Hayashi and Yoshida~(NIPS'16). Our approximation algorithm can be modified to handle the case where the minimization is done over a sphere. The analysis of our algorithms is obtained by combining results from graph limit theory,  along with a novel spectral decomposition of matrices. Specifically, we prove that a matrix $A$ can be decomposed into a structured part and a pseudorandom part, where the structured part is a block matrix with a polylogarithmic number of blocks, such that in each block all the entries are the same, and  the pseudorandom part has a small spectral norm, achieving better error bound than the existing decomposition theorem of Frieze and Kannan~(FOCS'96).
As an additional application of the decomposition theorem, we give a sublinear-time approximation algorithm for computing the top singular values of a matrix.

	\end{abstract}

\section{Introduction}

Quadratic function minimization/maximization is a versatile tool used in machine learning, statistics, and data mining and can represent many fundamental problems such as linear regression, $k$-means clustering,  principal component analysis (PCA), support vector machines, kernel machines and more (see~\cite{Murphy:2012}).
In general, quadratic function minimization/maximization is NP-Hard.
When the problem is convex (for minimization) or concave (for maximization), we can solve it by solving a system of linear equations, which requires $O(n^3)$ time, where $n$ is the number of variables. There are faster approximation methods based on  stochastic gradient descent~\cite{Bottou:2004},  and the multiplicative update algorithm~\cite{Clarkson:2012}.
However, these methods still require $\Omega(n)$ time, which is prohibitive when we need to handle a huge number of variables.

Quadratic function minimization over a sphere is also an important problem.
This minimization problem is often called the \emph{trust region subproblem} since it must be solved in each step of a \emph{trust region algorithm}. Trust region algorithms are among the most important tools in solving nonlinear programming problems, as they are robust and can be applied to ill-conditioned problems. In addition, trust region subproblems are useful in many other problems such as constrained eigenvalue problems~\cite{gander1989constrained}, least-square problems~\cite{zhang2010derivative}, combinatorial optimization problems~\cite{busygin2006new} and many more.
% In general trust region algorithms work iteratively in a way such that in each iteration, the algorithm solves (approximately or exactly) the trust region subproblem (for some choice of $r$) and picks its next step according to the value it obtains (See~\cite{yuan2015recent} for a more detailed information on trust region algorithms).
While the problem is non-convex, it has been shown that the problem exhibit strong duality properties and is known to be solved in polynomial time (see~\cite{ben1996hidden,ye2003new}). In particular, it was shown to be equivalent to some semidefinite programming optimization problems that can be solved in polynomial time (\cite{nesterov1994interior,alizadeh1995interior}). As in the non-constrained case, there are approximation algorithms based on gradient descent~\cite{nesterov1983method} and on reducing the problem to a sequence of eigenvalues computations~\cite{hazan2016linear}. However, as in the unconstrained case, these methods require running time which is linear in the number of the non-zero elements of the matrix (which might be linear in $n$).

\subsection{Our Contributions}
In this work, we provide sublinear-time approximation algorithms for minimizing quadratic functions, assuming random access to the entries of the input matrix and the vector.

First, we consider unconstraind minimization.
Specifically,
for a matrix $A\in \R^{n\times n}$ and vectors $\bmd,\bmb\in \R^{n}$, we consider the following quadratic function minimization problem:
\begin{align}
\min_{\bmv\in \R^n} \psi_{n,A,\bmd,\bmb}(\bmv),~\text{where } \psi_{n,A,\bmd,\bmb}(\bmv)= \dotprod{\bmv}{A\bmv}+n\dotprod{\bmv}{\diag(\bmd)\bmv}+n\dotprod{\bmb}{\bmv}\;.\label{Eq:MinDiscreat-intro}
\end{align}
Here $\diag(\bmd) \in \R^{n \times n}$ is a matrix whose diagonal entries are specified by $\bmd$ and $\dotprod{\cdot}{\cdot}$ denotes the standard inner product.

\begin{theorem}\label{thm:mainTheorem-intro}
	Fix $\eps>0$ and let $\bmv^*$ and $z^*$ be an optimal solution and the optimal value, respectively, of Problem~\eqref{Eq:MinDiscreat-intro}. Let $S$ be a random set such that each index $i \in \set{1,2,\ldots,n}$ is taken to $S$ independently w.p $k/n$ with
	$$k=\max\Bigg\{  O\left(\frac{\log^2 n}{\eps^2}\right), \left(\frac{1}{\eps}\right)^{O(1/\eps^2)}  \Bigg\}\;.$$ Then, the following holds with probability at least $2/3$:
	Let $\tilde{\bmv}^*$ and $\tilde{z}^*$ be an optimal solution and the optimal value, respectively, of the problem $$\min_{\bmv\in \R^{|S|}}\psi_{|S|,A|_{S},\bm d|_{S},\bm b|_{S}}(\bmv)\;,$$
	where $\cdot|_S$ is an operator that extracts a submatrix (or subvector) specified by an index set $S$.
	Then,
	$$\left|\frac{1}{|S|^2}\tilde{z}^*-\frac{1}{n^2}{z^*}\right|\le \eps L \max\left\{\frac{\normtwo{\tilde{\bmv}^*}^2}{|S|},\frac{\normtwo{\bmv^*}^2}{n}\right\}\;,
	$$
	where $L=\max\{\max_{i,j}|A_{ij}|,\max_i |d_i|,\max_i|b_i| \}$.
\end{theorem}

Recently, Hayashi and Yoshida~\cite{Hayashi:2016wh} proposed a constant-time sampling method for this problem with an additive error of $O(\epsilon L K_\infty^2 n^2)$ for $K_\infty = \max\set{\|\bmv^*\|_\infty,\|\tilde{\bmv}^*\|_\infty}$, where $\bmv^*$ and $\tilde{\bmv}^*$ are the optimal solutions to the original and sampled problems, respectively.
Although their algorithm runs in constant time, the guarantee is not meaningful when $K_\infty=\omega(1)$ because  the optimal value is always of order $O(L n^2)$.
Theorem~\ref{thm:mainTheorem-intro} shows that we can improve the additive error to $O(\epsilon L K^2_2 n^2)$, where $K_2 = \max\set{\|\bmv^*\|_2/\sqrt{n},\|\tilde{\bmv}^*\|_2/\sqrt{s}}$, as long as the number of samples $s$ is polylogarithmic (or more).
We note that we always have $K_2 \leq K_\infty$ and the difference is significant when $\bmv^*$ and $\tilde{\bmv}^*$ are sparse.
For example, if $\bmv^*$ and $\tilde{\bmv}^*$ have only $O(1)$ non-zero elements, then we have $K_2 = O(K_\infty / \sqrt{s})$.
Our new bound provides a trade off between the additive error and the time complexity, which was unclear from the argument by Hayashi and Yoshida~\cite{Hayashi:2016wh}.

Moreover, we consider minimization over a sphere.
Specifically, given a matrix $A\in \R ^{n\times n}$, vectors $\bmb,\bmd\in \R^n$ and $r > 0$, we consider the following quadratic function minimization problem over a sphere of radius $r$:
\begin{align}
\min_{\bmv : \normtwo{\bmv}\le r} \psi_{n,A,\bmd,\bmb}(\bmv)\;.\label{Eq:MinDiscreatConstrained-intro}
\end{align}

We give the first sublinear-time approximation algorithm for this problem. % \anote{(couldn't find any other sublinear variants)}. \ynote{(I think there's none.)}
\begin{theorem}\label{thm:mainTheoremSphere-intro}
	Let $\bmv^*$ and $z^*$ be an optimal solution and optimal value, respectively, of Problem~\eqref{Eq:MinDiscreatConstrained-intro}. Let $\eps>0$ and let $S$ be a random set such that each index $i \in \set{1,2,\ldots,n}$ is taken to $S$ independently w.p $k/n$ with
	$$k=\max\Bigg\{  O\left(\frac{\log^2 n}{\eps^2}\right), \left(\frac{1}{\eps}\right)^{O(1/\eps^2)}  \Bigg\}\;.$$
	Then, the following holds with probability at least $2/3$: Let $\tilde{\bmv}^*$ and $\tilde{z}^*$ be an optimal solution and the optimal value, respectively, of the problem $$\min_{\normtwo{v}\le\sqrt{\frac{|S|}{n}}r}\psi_{|S|,A|_{S},d|_S,b|_S}(\bmv)\;.$$
	Then,
	$$\left|\frac{1}{|S|^2}\tilde{z}^*-\frac{1}{n^2}{z^*}\right|\le \frac{\eps L r^2}{n}\;,$$
	where $L=\max\{\max_{i,j}|A_{ij}|,\max_i |d_i|,\max_i|b_i| \}$.
\end{theorem}
We can design a constant-time algorithm for~\eqref{Eq:MinDiscreatConstrained-intro} by using the result of~\cite{Hayashi:2016wh}, but the resulting error bound will be $O(\epsilon Lr^2)$, which is $n$ times worse than the bound in Theorem~\ref{thm:mainTheoremSphere-intro}.

%\anote{I feel like the following paragraph regrading singular values is a bit unrelated and disrupts the flow.... What do you think?} %\ynote{Added some connective sentences.}
The proofs of Theorems~\ref{thm:mainTheorem-intro} and~\ref{thm:mainTheoremSphere-intro} rely on a novel decomposition theorem of matrices, which will be discussed in Section~\ref{subsec:techniques}.

As another application of this decomposition theorem, we show that for any (small) $t$, we can approximate the $t$-th largest singular values of a matrix $A \in [-L,L]^{n \times m}$ (denoted $\sigma_t(A)$) to within an additive error of $O(L\sqrt{\epsilon t nm} )$ in time $$\max\Big\{  O\left(\frac{\max\{\log^2 n,\log^2 m\}}{\eps^2}\right), \left(\frac{1}{\eps}\right)^{O(1/\eps^2)}  \Big\}\;.$$
Our algorithms are very simple to implement, and do not require any structure in the input matrix. However, similar results (with better running time) can be obtained by applying known sampling techniques from~\cite{frieze2004fast}.
Formally, we prove the following.
\begin{theorem}\label{the:topEigen-intro}
	Given a matrix $A \in [-L,L]^{n \times m}$, $\epsilon \in (0,1) $, let $$k=
	\max\Bigg\{  O\left(\frac{\max\{\log^2 n,\log^2 m\}}{\eps^2}\right), \left(\frac{1}{\eps}\right)^{O(1/\eps^2)}  \Bigg\}\;.$$
	Then, for every $t = O(k)$, there is an algorithm that runs in $\poly(k)$ time, and outputs a value $z$ such that with probability at least $2/3$,
	$$
	| \sigma_t(A) - z | \leq L  \sqrt{\epsilon t nm}.
	$$
\end{theorem}
We note that since the $\sigma_t(A)\le L\sqrt{\frac{nm}{t}}$ (see Fact~\ref{Fact:singularvalues}), the \emph{relative} error the algorithm achieves is at least $\sqrt{\eps}\cdot t$. Therefore, to get meaningful approximation, one must have that $\sqrt{\eps}\cdot t<1$. So, if we wish to set $\eps=O(1)$ then we must have $t=O(1)$. 

Finlay, we present numerical experiments that confirm the empirical performance for accuracy and runtime of our singular values algorithm (see Section~\ref{sec:experiments})

\subsection{Related work}

In machine learning context, Clarkson \emph{et~al.}~\cite{Clarkson:2012} considered several machine learning optimization problems and gave sublinear-time approximation algorithms for those problems.
In particular, they considered approximate minimization of a quadratic function over the unit simplex $\Delta=\{\bmx\in \R^n\mid x_i\ge 0,\; \sum_i x_i=1\}$.
Namely, given a positive semidefinite matrix $A\in \R^{n\times n}$ and $\bmb\in \R^n$, they showed that it is possible to obtain an approximate solution to $\min_{x\in \Delta}\bmx^\top A\bmx+\bmx^\top \bmb$ (up to an additive error of $\eps$) in $\widetilde{O}(n/\eps^2)$ time, which is sublinear in the input size $\Theta(n^2)$.
In contrast, our algorithms run in polylogarithmic time and are much more efficient.
Hayashi and Yoshida~\cite{NIPS2017_6841} proposed a constant-time approximation algorithm for Tucker decomposition of tensors, which can be seen as minimizing low-degree polynomials.

In addition to the work of Hayashi and Yoshida~\cite{Hayashi:2016wh} mentioned above, an additional line of relevant work is constant-time approximation algorithms for the max cut problem on dense graphs~\cite{Frieze:1996er,Mathieu:2008vs}.
Let $L_G \in \R^{n \times n}$ be the Laplacian matrix of a graph $G$ on $n$ vertices.
Then, the max cut problem can be seen as maximizing $\langle \bmx,L_G \bmx\rangle$ subject to $x_i \in \set{-1/\sqrt{n},1/\sqrt{n}}$, and these methods approximate the optimal value to within $O(\epsilon n)$.
Our method for approximating the largest singular values can be seen as an extension of these methods to a continuous setting.

\subsection{Techniques}\label{subsec:techniques}
The main ingredient in our proof is a novel \emph{spectral decomposition theorem} of matrices, which may be of independent interest.
The theorem states that we can decompose a matrix $A \in \R^{n \times m}$ into a structured matrix $A^{\text{str}} \in \R^{n \times m}$ and a pseudorandom matrix $A^{\text{psd}} \in \R^{n \times m}$.
Here, $A^{\text{str}}$ is structured in the sense that it is a block matrix with a polylogarithmic number of blocks such that the entries in each block are equal.
Also, $A^{\text{psd}}$ is pseudorandom in the sense that it has a small spectral norm. Formally, we prove the following.
For a matrix $A \in \R^{n \times m}$, we define its max norm as $\|A\|_{\max}=\max_{ 1 \leq i \leq n}\max_{1 \leq j \leq m}|A_{ij}|$.
\begin{theorem}\label{thm:MatrixDecomposition-intro}
	For any matrix $A\in[-L,L]^{n\times m }$ and $\gamma \in (0,1)$, there exists a decomposition $A= A^{\text{str}}+A^{\text{psd}}$ with the following properties for $N = \sqrt{nm}$:
	\begin{enumerate}
		\item $A^{\text{str}}$ is \emph{structured} in the sense that it is a block matrix with $\left(\frac{1}{\gamma}\right)^{O(1/\gamma^2)}$ blocks, such that the entries in each block are equal.
		\item $\|A^{\text{psd}}\|_2 \le \gamma NL$.
		\item $\|A^{\text{str}}\|_{\max} = L/\gamma^{O(1)}$.
	\end{enumerate}
\end{theorem}
Our decomposition theorem is a strengthening of the matrix decomposition result of Frieze and Kannan~\cite{frieze1999quick,Frieze:1996er}. In particular, they showed that any matrix $A\in\R^{n\times m}$ can be decomposed to $D_1+\cdots+D_s+W$ for $s=O(1/\gamma^2)$, where the matrices $D_i$ are block matrices and $\norm{W}_C\le \gamma nm \|A\|_{\max}$.
Here, $\|W\|_C$ is the cut norm, which is defined as $$\max\limits_{S \subseteq \set{1,\ldots,n}}\max\limits_{T\subseteq \set{1,\ldots,m}}\left|\sum_{i\in S}\sum_{j\in T}W_{ij}\right|\;.$$
By using a result of Nikiforov~\cite{nikiforov2009cut} that $\|W\|_2=O(\sqrt{nm\cdot\|W\|_{\max}\cdot \|W\|_C})$ and the fact that the Frieze-Kannan result implies $\|W\|_{\max}\le\sqrt{s}\|A\|_{\max} $, we get that $\normtwo{W}= O(nm\cdot \|A\|_{\max})$, which is too loose, and thus insufficient for our applications.

Given our decomposition theorem, we can conclude the following.
When approximating~\eqref{Eq:MinDiscreat-intro} and~\eqref{Eq:MinDiscreatConstrained-intro}, we can disregard the pseudorandom part $A^{\text{psd}}$.
This will not affect our approximation by much, since $A^{\text{psd}}$ has a small spectral norm.
In addition, as $A^{\text{str}}$ consists of a polylogarithmic number of blocks, such that the entries in each block are equal, we can hit all the blocks by sampling a polylogarithmic number of indices.
Hence, we can expect that $A|_S$ is a good approximation to $A$.
To formally define the distance between $A$ and $A|_S$ and to show it is small, we exploit graph limit theory, initiated by \Lovasz and Szegedy~\cite{Lovasz:2006jj} (refer to~\cite{Lovasz:2012wn} for a book).

\section{Preliminaries}\label{sec:pre}
For an integer $n$ we let $[n]\eqdef\{1,2,\ldots,n\}$.
Given a set of indices $S=\{i_1,\ldots,i_k\}$, and a vector $\bmv\in\R^n$, we let $\bmv|_S \in \R^k$ be the \textit{restriction} of $\bmv$ to $S$; that is, $(\bmv|_S)_j=v_{i_j}$, for every $i\in[k]$.
Similarly, for a matrix $A\in\R^{n\times m}$ and sets $S_R=\{i_1,\ldots,i_{k_R}\}\subseteq [n]$ and $S_C=\{j_1,\ldots,j_{k_C}\}\subseteq[m]$, we denote the \textit{restriction} of $A$ to $S_R\times S_C$ by $A|_{S_R\times S_C} \in \R^{k_R \times k_C}$; that is, $(A|_{S_R\times S_C})_{j\ell}=A_{i_ji_\ell}$, for every $j\in [k_R]$ and $\ell\in[k_C]$.
When $S_R=S_C=S$ we often use $A|_{S}$ as a shorthand for $A|_{S\times S}$.
We use the notation $x=y\pm z$ as a shorthand for $y-z\le x\le y+z$.

\noindent Given a matrix $A\in\R^{n\times m}$ we define the \textit{Frobenius norm} of $A$ as $\normF{A}=\sqrt{\sum_{(i,j)\in[n]\times[m]}A^2_{ij}}$ and the \textit{max norm} of $A$ as $\norm{A}_{\max}=\max_{i \in [n],j \in [m]}|A_{ij}|$.
For a matrix $A \in \R^{n \times m}$, we let $\sigma_\ell(A)$ denote the $\ell$-th largest singular value of $A$.
It is well known that the largest singular value can be evaluated using the following.
$$\sigma_1(A)=\max_{\bmv \in \R^m:\; \normtwo{\bmv}\le 1}\normtwo{A\bmv}\;.$$
In addition, we state the following fact regarding the singular values.
\begin{fact}\label{Fact:singularvalues}
	Let $A\in\R^{n\times m}$, and consider the singular values of $A$: $\sigma_1(A)\ge \ldots\ge\sigma_{\min\set{n,m}}(A) $. Then, for every $1\le\ell\le\min\set{n,m}$, $\sigma_\ell(A)\le\frac{\normF{A}}{\sqrt{\ell}}$.
\end{fact}

\section{Spectral Decomposition Theorem}\label{sec:decomposition}
In this section we will prove the following decomposition theorem.
\begin{theorem}[Spectral decomposition, restatement of Theorem~\ref{thm:MatrixDecomposition-intro}]\label{thm:MatrixDecomposition}
	For any matrix $A\in[-L,L]^{n\times m }$ and $\gamma \in (0,1)$, there exists a decomposition $A= A^{\mathrm{str}}+A^{\mathrm{psd}}$ with the following properties for $N = \sqrt{nm}$:
	\begin{enumerate}
		\item $A^{\mathrm{str}}$ is \emph{structured} in the sense that it is a block matrix with $O\left(\left(\frac{1}{\gamma^{10}}\right)^{3/\gamma^2}\right)$ blocks, such that the entries in each block are equal.
		\item $\|A^{\mathrm{psd}}\|_2\le 7\gamma NL$.
		\item $\|A^{\mathrm{str}}\|_{\max}\le \frac{2}{\gamma^{11}}L$.
	\end{enumerate}
\end{theorem}

The above theorem will serve as a central tool in the analysis of our algorithms.
The fact that $A^{\mathrm{str}}$ is a block matrix with polylogarithmic number of blocks, such that the entries in each block are equal, implies that by using polylogarithmic number of samples, we can query (with high probability) an entry from each of the blocks. In addition, the fact that $A^{\mathrm{psd}}$ has a small spectral norm allows us to disregard it, which only paying a small cost in the error of our approximation.

In order to prove the theorem, we introduce the following definition, two lemmas and a claim.
\begin{definition}
	We say that a partition $\mathcal{Q}$ is a \textit{refinement} of a partition $\mathcal{P}=\{V_1,\ldots,V_p\}$, if $\mathcal{Q}$ is obtained from $\mathcal{P}$, by splitting some sets $V_i$ into one or more parts.
\end{definition}
%We will use the above definition when proving the following lemma.
\begin{lemma}\label{lem:GoodBlockApprox}
	Given a matrix $A\in[-L,L]^{n\times m}$ and $\gamma \in(0,1)$, there exists a block matrix $A^{\mathrm{str}} \in \R^{n \times m}$ with $O\left(\left(\frac{1}{\gamma^{10}}\right)^{3/\gamma^2}\right)$ blocks such that the entries in each block are equal and  $\normtwo{A-A^{\mathrm{str}}}\le 7\gamma N L$, where $N = \sqrt{nm}$.
\end{lemma}
In order to prove Lemma~\ref{lem:GoodBlockApprox}, we will need to prove the following.
\begin{lemma}\label{lem:BoundValueApproxMarix}
	Given a matrix $A\in [-L,L]^{n\times m}$ and $\gamma \in (0,1)$, let $A'=\sum_{\ell:\sigma_\ell\ge\gamma NL}\sigma_\ell \bmu^\ell(\bmv^\ell)^\top$. Then, $\|A'\|_{\max}\le \frac{L}{\gamma^3}$.
\end{lemma}
\begin{proof}
	Assume that $A$ has $s$ singular values such that $\sigma_\ell \ge\gamma NL$.
	For any $\ell \in [s]$, let $M_\ell=\sigma_\ell \bmu^\ell (\bmv^\ell)^\top$, and let $B$ denote $\sum_{\ell:\;\sigma_\ell<\gamma NL}\sigma_\ell \bmu^\ell (\bmv^\ell)^\top$.
	Then, we can write $A$ as,
	$$A=\underbrace{\sigma_1 \bmu^1 (\bmv^1)^\top+\cdots+\sigma_s \bmu^s (\bmv^s)^\top }_{A'}+\sum_{\ell:\;\sigma_\ell<\gamma NL}\sigma_\ell \bmu^\ell (\bmv^\ell)^\top=M_1+\cdots+M_s +B\;.$$
	Consider any $\ell \in[s]$, $(i,j)\in [n]\times[m]$, and let $\beta^\ell_{ij}\eqdef|\sigma_\ell u^\ell_i v^\ell_j|$.
	%	Assume, without the loss of generality, that $|u^t_i|\le|v^t_j|$ (otherwise, in what follows, we do the same for the rows)  .
	%	Therefore, we have that $|v^t_i|\ge \sqrt{\frac{\beta^t_{ij}}{\sigma_t}}$ (or otherwise, $|\sigma_t u^t_i v^t_j|<\beta^t_{ij}$).
	
	Let $c^A_j$ denote the $j$-th column  of the matrix $A$.
	So, $c^A_j=c_j^{M_1}+\cdots+c_j^{M_s}+c_j^{B}$. For any $\ell \in[s]$, $c_j^{M_\ell}$ is perpendicular to $c_j^B$ and all $\{c_j^{M_t}\}_t$ such that $t \neq \ell$, and thus,
	$$\normtwo{c_j^A}\ge\normtwo{c_j^{M_\ell}}=\normtwo{\sigma_\ell \bmu^\ell v^\ell_j}=\sigma_\ell  |v^\ell_j|\;.$$
	Let $r^A_i$ denote the $i$-th row of the matrix $A$.
	Then similarly, we have $\normtwo{r_i^A} \ge \sigma_\ell  |u^\ell_i|$, and it follows that $\normtwo{c_j^A} \normtwo{r_i^A} \ge \sigma_\ell \beta^\ell_{ij}$.
	
	On the other hand, since $A\in [-L,L]^{n\times m}$, we have $\normtwo{c_j^A}\le \sqrt{n}L$ and $\normtwo{r_i^A}\le \sqrt{m}L$, and therefore,
	$$\beta^\ell_{ij}\le \frac{\sqrt{nm}L^2}{\sigma_\ell}\le \frac{\sqrt{nm}L^2}{\gamma NL}=\frac{L}{\gamma}\;.$$
	By the fact that $s\le \frac{1}{\gamma^2}$, we get that $\|A'\|_{\max}\le \frac{L}{\gamma^3}$, which concludes the proof.
\end{proof}

With this lemma at hand, we are ready to prove Lemma~\ref{lem:GoodBlockApprox}.

\begin{proofof}{Lemma~\ref{lem:GoodBlockApprox}}
	Recall that $A$ can be written as,
	$$
	A=\sum_\ell \sigma_\ell \bmu^\ell(\bmv^\ell)^\top\;,
	$$
	where $\sigma_1\ge \cdots\ge\sigma_{\min\set{n,m}} \geq 0 $ are the singular values of $A$ and $\bmu^1,\ldots,\bmu^{\min\set{n,m}} \in \R^n$ and $\bmv^1,\ldots,\bmv^{\min\set{n,m}} \in \R^m$ are the corresponding left and right singular vectors.
	If we let $A'$ be such that
	$$
	A'=\sum_{\ell: \; \sigma_\ell\ge\gamma NL}\sigma_\ell \bmu^\ell(\bmv^\ell)^\top\;,
	$$
	then we have that $\frac{\|(A-A')\bmx\|_2}{\|x\|_2}\le \gamma NL$ for any $\bmx\in \R^m$.
	
	Next we show the existence of  $A^{\mathrm{str}}$, which is a block matrix (with $O\left(\left(\frac{1}{\gamma^{10}}\right)^{3/\gamma^2}\right)$ blocks), such that $A^{\mathrm{str}}$ has the same value on every block.
	We construct $A^{\mathrm{str}}$ as follows.
	
	Let $\eps=\eps(\gamma)$ be determined later and let $J=\frac{1}{\gamma^2}$.
	Let $\delta_n\eqdef\frac{\eps}{J\sqrt{ n}}$ and $T\eqdef{(\frac{J}{\eps})}^{3/2}$, and partition the interval $[0,1]$ into buckets $B^n_1,\ldots,B^n_T,B_{\mathrm{Large}}^n$
	%\ynote{($B^n_1,\ldots,B^n_T,B_{\mathrm{Large}}^n$ do not form a partition of $[0,1]$. Also, the symbol $L$ is reserved for the upper bound on the absolute value of an element.) }
	such that
	$$
	B^n_t \eqdef[(t-1)\delta_n,t\delta_n)\;\;\forall t\in [T]\qquad \mathrm{and}\qquad B^n_{\mathrm{Large}}\eqdef[\sqrt{J/\eps n},1]\;.
	$$
	For every $\bmu^\ell$ such that $\sigma_\ell\ge \gamma NL$, we define a partition $\mathcal{P}_R^\ell=\{P^\ell_{R,1}\ldots,P^\ell_{R,T},\Sigma^\ell_R\}$ of the indices in $[n]$ so that $P^\ell_{R,t} = \{i\in [n]\mid |u^\ell_i|\in B_t^n\}$ for each $t \in [T]$ and $\Sigma^\ell_R=\{i\in[n]\mid |u^\ell_i|\in B^n_{\mathrm{Large}}\}$.
	We eliminate emptysets from $\mathcal{P}_R^\ell$ if exist.
	%	For every $1\le t \le T$ define $P_t(\bmu^\ell)=\{i\in [n]\mid |u^\ell_i|\in B_t^n\}$, $\Sigma^\ell_n=\{i\in[n]\mid |u^\ell_i|\in B^n_{\mathrm{Large}}\}$, and let $\mathcal{P}(\bmu^\ell)=\{P_t(\bmu^\ell)\mid P_t(\bmu^\ell)\neq\emptyset\}\cup \Sigma^\ell_n$ \ynote{very minor: do we want to include $\Sigma^\ell_n$ even when it is $\emptyset$?}.
	Note that by the definition of $\Sigma^\ell_R$ we have that $|\Sigma^\ell_R|\le \eps n/J$.
	Next, for every $\ell$ such that $\sigma_\ell\ge \gamma NL$, we define $\hat{\bmu}^\ell$ as follows.
	$$
	\hat{u}^\ell_i=\begin{cases}
	0, &\text{if  } |u^\ell_i|\in B^n_{\mathrm{Large}}\;, \\
	\delta_n(t-1), & \text{if  } |u^\ell_i|\in B^n_t \text{ for some } t \in [T]\;.
	\end{cases}\;.
	$$
	Next, let $\Sigma_R=\bigcup_{\ell: \sigma_\ell\ge \gamma NL}\Sigma_R^\ell$ and let $\mathcal{P}_R$ be a partition of $[n]\setminus \Sigma_R$ that refines all $\{P^\ell_{R,t} \mid P^\ell_{R,t} \neq \emptyset\}$ for which $\ell$ is such that $\sigma_\ell \ge \gamma NL$.
	
	\medskip
	
	\noindent Similarly define $\hat{\bmv}^\ell$ from $\bmv^\ell$ by setting $\delta_m = \frac{\epsilon}{J\sqrt{m}}$, and defining $\mathcal{P}^\ell_C=\{P_{C,1}^\ell,\ldots,P^\ell_{C,T},\Sigma_C^\ell\}$ analogously for each $\ell$.
	Let $\Sigma_C=\bigcup_{\ell: \sigma_\ell \ge \gamma NL}\Sigma_C^\ell$ and let $\mathcal{P}_C$ be a partition of $[m]\setminus \Sigma_C$ that refines all $\{P^\ell_{C,t}\mid P^\ell_{C,t}\neq\emptyset\}$ for which $\ell$ is such that $\sigma_\ell\ge \gamma NL$.
	
	Define,
	$$
	A^{\mathrm{str}}\eqdef\sum_{\ell: \;\sigma_\ell \ge\gamma NL}\sigma_\ell \hat{\bmu}^\ell(\hat{\bmv}^\ell)^\top\; \qquad \mathrm{and} \qquad A^{\mathrm{psd}}\eqdef A - A^{\mathrm{str}}.
	$$
	
	By Fact~\ref{Fact:singularvalues} we have that $\sigma_\ell \le \frac{NL}{\sqrt{\ell}}$ for all $\ell\ge 1$.
	Therefore, we can have at most $J=1/\gamma^2$ indices such that $\sigma_\ell \ge \gamma NL$.
	Thus, the partition $\mathcal{P}_R$ satisfies $|\mathcal{P}_R| \leq \left(\left(\frac{J}{\eps}\right)^{\frac{3}{2\gamma^2}}\right)=\left(\frac{1}{\gamma^2\eps}\right)^{(3/2\gamma^2)}$.
	Similarly, we have $|\mathcal{P}_C| \leq \left(\frac{1}{\gamma^2\eps}\right)^{(3/2\gamma^2)}$.
	Therefore, the resulting matrix is a block matrix with $O\left(\left(\frac{1}{ \gamma^2\eps}\right)^{\frac{3}{\gamma^2}}\right)$ many blocks, such that all the entries in each block are the same.
	We refer to $\mathcal{P}_R$ and $\mathcal{P}_C$ the \textit{row partition of $A^{\mathrm{str}}$} and the \textit{column partition of $A^{\mathrm{str}}$} respectively.
	
	Next, we have that
	\begin{align*}
	\|A^{\mathrm{psd}}\bmx\|=	\|(A-{A}^{\mathrm{str}})\bmx\|_2^2&\le \normtwo{(A-A')\bmx}^2+\normtwo{(A'-{A}^{\mathrm{str}})\bmx}^2\\
	&\le\gamma^2 N^2L^2 \normtwo{\bmx}^2+ \normtwo{(A'-{A}^{\mathrm{str}})\bmx}^2 \;.
	\end{align*}

	\noindent Consider the ``error'' term $\normtwo{(A'-A^{\mathrm{str}})\bmx}^2$.
	\begin{align*}
	& \normtwo{(A'-A^{\mathrm{str}})\bmx}^2 \le \normtwo{\bmx}^2\sum_{(i,j)\in [n]\times [m]}\left(A'_{ij}-A^{\mathrm{str}}_{ij}\right)^2= \normtwo{\bmx}^2\Bigg(\sum_{(i,j)\in \bar{\Sigma}_R \times \bar{\Sigma}_C}\left(A'_{ij}-A^{\mathrm{str}}_{ij}\right)^2\\
	&\qquad\;\;\;\qquad\qquad+\sum_{(i,j)\in (\Sigma_R \times \bar{\Sigma}_C)\cup (\bar{\Sigma}_R\times \Sigma_C)}\left(A'_{ij}-A^{\mathrm{str}}_{ij}\right)^2+\sum_{(i,j)\in \Sigma_R \times \Sigma_C}\left(A'_{ij}-A^{\mathrm{str}}_{ij}\right)^2\Bigg),
	\end{align*}
	where $\bar{\Sigma}_R = [n] \setminus \Sigma_R$ and $\bar{\Sigma}_C = [m] \setminus \Sigma_C$.
	
	\noindent We analyze each of these terms separately.
	First, note that $$\left(A'_{ij}-A^{\mathrm{str}}_{ij}\right)^2=\left(\sum\limits_{\ell:\sigma_\ell\ge \gamma NL }\sigma_\ell(u^\ell_i v^\ell_j-\hat{u}^\ell_i\hat{v}^\ell_j)\right)^2\;.$$
	So,
	\begin{align*}
	\sum_{(i,j)\in \bar{\Sigma}_R \times \bar{\Sigma}_C }\left(A'_{ij}-A^{\mathrm{str}}_{ij}\right)^2&=\sum_{(i,j)\in \bar{\Sigma}_R \times \bar{\Sigma}_C }\left(\sum_{\ell:\sigma_\ell\ge \gamma NL }\sigma_\ell(u^\ell_iv^\ell_j-\hat{u}^\ell_i\hat{v}^\ell_j)\right)^2\\
	&\stackrel{(\ast)}{\le}\sum_{(i,j)\in \bar{\Sigma}_R \times \bar{\Sigma}_C }\left(\frac{2\sqrt{\eps}}{N}\sum_{\ell:\sigma_\ell\ge \gamma NL }\sigma_\ell\right)^2\\
	&\stackrel{(**)}{\le} \sum_{(i,j)\in \bar{\Sigma}_R \times \bar{\Sigma}_C } \left(\frac{2\sqrt{\eps}}{N}\cdot\frac{2NL}{\gamma}\right)^2 \le \frac{16\eps N^2L^2}{\gamma^2}\;.
	\end{align*}
	Here, $(*)$ follows from the fact that for two indices $i\in\bar{\Sigma}_R, j\in \bar{\Sigma}_C$ we have that $\hat{u}^\ell_i \hat{v}^\ell_j=(u^\ell_i\pm \delta_n)(v^\ell_j\pm \delta_m)\le u^\ell_iv^\ell_j \pm\frac{2\sqrt{\eps}}{N}$. %\ynote{(I can't verify because the definition of $B^n_1,\ldots,B^n_T,B^n_L$ is unclear.)}
	$(**)$ follows from the fact that there can be at most $1/\gamma^2$ indices $\ell$, such that $\sigma_\ell\ge\gamma NL $, and therefore
	\begin{align}
	\sum_{\ell: \sigma_\ell \ge \gamma NL}\sigma_\ell\le \sum_{\ell=1}^{1/\gamma^2}\sigma_\ell \le\sum_{\ell=1}^{1/\gamma^2}\frac{NL}{\sqrt{\ell}}\le \frac{2NL}{\gamma} \;.	 \label{eq:sum-of-large-singular-values}
	\end{align}
	
	\noindent Next, we have that,
	\begin{align*}
	& \sum_{(i,j)\in (\Sigma_R \times \bar{\Sigma}_C)\cup (\bar{\Sigma}_R\times \Sigma_C)}\left(A'_{ij}-A^{\mathrm{str}}_{ij}\right)^2\\
	&=\sum_{(i,j)\in (\Sigma_R \times \bar{\Sigma}_C)\cup (\bar{\Sigma}_R\times \Sigma_C)}\left(\sum_{\ell:\sigma_\ell\ge \gamma NL }\sigma_\ell(u^\ell_iv^\ell_j-\hat{u}^\ell_i\hat{v}^\ell_j)\right)^2\\
	&\stackrel{(***)}{=}\sum_{(i,j)\in (\Sigma_R \times \bar{\Sigma}_C)\cup (\bar{\Sigma}_R\times \Sigma_C)}\left(A'_{ij}\right)^2\le \sum_{(i,j)\in (\Sigma_R \times \bar{\Sigma}_C)\cup (\bar{\Sigma}_R\times \Sigma_C)}\left(\frac{L}{\gamma^3}\right)^2\le \frac{2\eps N^2 L^2}{\gamma^6}\;.
	\end{align*}
	Here, $(***)$ follows from the fact that when one of the indices is in $\Sigma_R$ or $\Sigma_C$ we set the corresponding entry in the rounded vector to $0$.
	In addition, the last inequality follows from the fact that when we remove the lower singular part of the matrix, we can only increase the value by at most factor of $1/\gamma ^3$ (see Lemma \ref{lem:BoundValueApproxMarix}).
	Finally,
	\begin{align*}
	\sum_{(i,j)\in \Sigma_R \times \Sigma_C}\left(A'_{ij}-A^{\mathrm{str}}_{ij}\right)^2&=\sum_{(i,j)\in \Sigma_R \times \Sigma_C}\left(\sum_{\ell:\sigma_\ell\ge \gamma NL }\sigma_\ell(u^\ell_iv^\ell_j-\hat{u}^\ell_i\hat{v}^\ell_j)\right)^2\\
	&\stackrel{}{=}\sum_{(i,j)\in \Sigma_R \times \Sigma_C}\left(A'_{ij}\right)^2\le \frac{\eps^2 N^2L^2}{\gamma^6}
	\end{align*}
	
	Combining all the three terms and setting $ \eps=\gamma^8$ gives,
	$$
	\normtwo{(A'-A^{\mathrm{str}})\bmx}^2 \le \normtwo{\bmx}^2\left(\frac{16\eps}{\gamma^2}+\frac{2\eps}{\gamma^6}+\frac{\eps^2}{\gamma^6}\right)N^2L^2\le 19\gamma^2\normtwo{\bmx}^2N^2L^2.
	$$
	Therefore, we get that,
	\begin{align*}
	\|(A-A^{\mathrm{str}})\bmx\|_2^2&\le 2\normtwo{(A-A')\bmx}^2+2\normtwo{(A'-A^{\mathrm{str}})\bmx}^2 \\
	&\le 2(\gamma^2 N^2L^2 \normtwo{\bmx}^2+ 19\gamma^2\normtwo{\bmx}^2N^2L^2)\le 40\gamma^2\normtwo{\bmx}^2N^2L^2\;,
	\end{align*}
	
	and the lemma follows.
\end{proofof}

We are left with bounding the max norm of $A^{\text{str}}$.
\begin{claim}\label{lem:BoundedValuesSTR}
	Given a matrix $A\in [-L,L]^{n\times m}$ and $\gamma \in (0,1)$, let $ A^{\mathrm{str}}\in \R^{n \times m}$ be the block approximation matrix defined above.
	Then, $\|{A}^{\mathrm{str}}\|_{\max}\le \frac{2L}{\gamma^{11}}$
\end{claim}

\medskip
\begin{proof}%{Claim~\ref{lem:BoundedValuesSTR}}
	By the definition of $A^{\mathrm{str}}$ we have that  $| A^{\mathrm{str}}_{ij}|=|\sum_{\ell : \sigma_\ell\ge \gamma NL}\sigma_\ell \hat u^\ell_i \hat v^\ell_j|$. By the definition of the rounding process, we have that for every $i\in [n]$ we have that $|\hat u^\ell_i|\le\sqrt{\frac{J}{\eps n}}= \frac{1}{\gamma^5 \sqrt{n}}$ (recall that $J=1/\gamma^2$ and $\eps=\gamma^8$). % \ynote{(how did you get this?)}.
	Similarly, for every $j\in[m]$ we have that $|\hat v^\ell_j|\le \frac{1}{\gamma^5 \sqrt{m}}$.
	Therefore, $$
	|A^{\mathrm{str}}_{ij}|=\left|\sum_{\ell : \sigma_\ell\ge \gamma NL}\sigma_\ell \hat u^\ell_i \hat v^\ell_j\right|\le \frac{1}{\gamma^{10 }N}\sum_{\ell : \sigma_\ell\ge \gamma NL}\sigma_\ell\le \frac{2L}{\gamma^{11}}\;,
	$$
	where the last inequality uses~\eqref{eq:sum-of-large-singular-values}.
\end{proof}

\begin{proofof} {Theorem~\ref{thm:MatrixDecomposition}}
	The proof follows directly from Lemma~\ref{lem:GoodBlockApprox} and Claim~\ref{lem:BoundedValuesSTR}.
\end{proofof}

\section{Dikernels and Sampling Lemmas}\label{sec:sampling}
In this section we will formalize the idea that $A|_{S_R\times S_C}$ is a good approximation of $A$ when $S_R$ and $S_C$ are uniformly random subsets of indices.
(The proof for $A|_S$, where $S$ is uniformly random subset of indices, is almost identical and we omit it.)
We start by providing some background on dikernels and their connection to matrices and then move on to proving our sampling lemmas.

\subsection{Dikernels and Matrices}
We call a (measurable) function $f\colon[0,1]^2\to \R$ a \textit{dikernel}.
We can regard a dikernel as a matrix whose index is specified by a real value in $[0, 1]$.
For two functions $f, g : [0, 1] \to\R$, we define their inner product as $\langle f,g\rangle\eqdef\int_{0}^{1}f(x)g(x)dx$.
For a dikernel $\mathcal{A} :[0,1]^2 \to \R$ and a function $f :[0,1]\to\R$, we define the function $\mathcal{A} f\colon[0,1]\to \R$ as $(\mathcal{A} f)(x)=\dotprod{\mathcal{A} (x,\cdot)}{f}$.
In addition, we define the \textit{spectral norm} of $\mathcal{A} $ as $\normtwo{\mathcal{A} }\eqdef \sup_{f\colon[0,1]\to \R}\frac{\normtwo{\mathcal{A} f}}{\normtwo{f}}$, and  the \textit{Frobenius norm} of $\mathcal{A}$ as $\normF{\mathcal{A}}\eqdef\sqrt{\int_{0}^{1}\int_{0}^{1}\mathcal{A}(x,y)^2 dxdy }$.

For an integer $n \in \mathbb{N}$, let $I^n_1=[0,\frac{1}{n}]$, and for every $1<k\le n$, let $I^n_k=(\frac{k-1}{n},\frac{k}{n}]$.
For $x\in[0,1]$, we define $i^n(x) $ as the unique integer $k\in [n]$ such that $x\in I^n_k$.
\begin{definition}
	Given a matrix $A\in\R^{n\times m}$, we construct the corresponding dikernel $\mathcal{A}$ as $\mathcal{A}(x,y)=A_{i^n(x),i^m(y)}$.
	In addition, given two sets of indices $S_R \subseteq [n]$ and $S_C \subseteq [m]$, when we write $\mathcal{A|_{S_\text{R}\times S_\text{C}}}$, we first extract the matrix $A|_{S_R\times S_C}$ and then consider its corresponding dikernel.
\end{definition}
The following lemma shows that the spectral norms of $A$ and $\mathcal{A}$ are essentially the same up to normalization. 
\begin{lemma}\label{lem:normRelation}
	Let $A\in \R^{n\times m}$ be a matrix.
	Then, we have
	$$
	\max_{\bmv \in \R^m} \frac{\normtwo{A\bmv}^2}{\normtwo{\bmv}^2}= nm\cdot \sup_{f\colon[0,1]\to \R} \frac{\normtwo{\mathcal{A}f}^2}{\normtwo{f}^2}\;.
	$$ % \ynote{The constraints $\normtwo{\bmv} \le L$ and $\normtwo{f} \le \frac{L}{\sqrt{m}}$ are confusing because the optimal values do not change without the constraint (as long as $\bmv \neq {\bm 0}$ and $f \not \equiv 0 $).}
\end{lemma}
\begin{proof} %[Proof of Lemma~\ref{lem:normRelation}]
	Before starting the proof we introduce some notations and an important observation. We note that $\sup_{f:[0,1]\to \R} \frac{\normtwo{\mathcal{A}f}^2}{\normtwo{f}^2}$ has a minimizer, since the objective function is weakly continuous and we may assume that $\normtwo{f}\le 1$ which is weakly compact.
	%	For every $i \in[n]$, we let $i_n(I_i)$ denote the unique integer which correspond to the interval $I_i$ in $A$, and in particular, we have $i_n(I_i)=i$.
	%	Similarly, we define $i_m(I^m_j)$ for $j \in [m]$.
	For every $x\in[0,1]$ and $j\in[m]$, we let $\mathcal{A}(x,I^m_j)=A_{i^n(x)j}$.
	Also for every $(i,j) \in [n] \times [m]$, we let $\mathcal{A}(I^n_i,I^m_j)=A_{ij}$.
	
	We start by showing that $\max_{\bmv \in \R^m} \frac{\normtwo{A\bmv}^2}{\normtwo{\bmv}^2}\le nm\cdot \sup_{f:[0,1] \to \R}\frac{\normtwo{\mathcal{A}f}^2}{\normtwo{f}^2}$.
	Given a vector $\bmv \in \R^m$, we define the function $f:[0,1]\to\R$ as $f(x)=v_{i^m(x)}$.
	Then,
	\begin{align*}
	\normtwo{\mathcal{A}f}^2&=\int_{0}^{1}\left(\int_{0}^{1}\mathcal{A}(x,y)f(y)dy\right)^2dx=\int_{0}^{1}\left(\sum_{j\in [m]}\int_{I^m_j}\mathcal{A}(x,y)f(y)dy\right)^2dx\\
	&=\int_{0}^{1}\left(\frac{1}{m}\sum_{j\in [m]}\mathcal{A}({x},I^m_j)v_j\right)^2dx=\frac{1}{n}\sum_{i \in [n]}\left(\frac{1}{m}\sum_{j\in [m]}\mathcal{A}({I^n_i},I^m_j)v_j\right)^2 \\
	&=\frac{1}{n}\sum_{i \in [n]}\left(\frac{1}{m}\sum_{j\in [m]}{A}_{ij}v_j\right)^2=\frac{1}{nm^2}\normtwo{A\bmv}^2\;.
	\end{align*}
	%\ynote{Define $\hat{A}(x,I_j)$}
	%  \amargin{Added definition of $\hat{A}(x,I_j)$ }
	In addition,
	$$
	\normtwo{f}^2=\int_{0}^{1}f(x)^2dx=\sum_{j\in[m]}\int_{I^m_j}f(x)^2dx=\frac{1}{m}\sum_{j\in[m]}v_j^2=\frac{1}{m}\normtwo{\bmv}^2\;.
	$$
	%	Note that $\normtwo{f}\le \frac{L}{\sqrt{m}}$ and thus, $\max_{\bmv:\; \normtwo{\bmv}\le L} \frac{\normtwo{A\bmv}^2}{\normtwo{\bmv}^2}\le nm\cdot \sup_{f:[0,1] \to \R}\frac{\normtwo{\mathcal{A}f}^2}{\normtwo{f}^2}$.
	
	Next, we will prove that $\max_{\bmv \in \R^m} \frac{\normtwo{A\bmv}^2}{\normtwo{\bmv}^2}\ge nm\cdot \sup_{f:[0,1] \to \R}\frac{\normtwo{\mathcal{A}f}^2}{\normtwo{f}^2}$.
	Let $f:[0,1] \to \R$ be a measurable function. Then, for any $x\in[0,1]$, consider the partial derivative,
	\begin{align*}
	\frac{\partial}{\partial f(x)}\frac{\normtwo{\mathcal{A}f}^2}{\normtwo{f}^2}=\frac{\normtwo{f}^2\cdot\frac{\partial}{\partial f(x)}\normtwo{\mathcal{A}f}^2-\normtwo{\mathcal{A}f}^2\cdot\frac{\partial}{\partial f(x)}\normtwo{f}^2}{\normtwo{f}^4}\;.
	\end{align*}
	Note that,
	$$
	\frac{\partial}{\partial f(x)}{\normtwo{\mathcal{A}f}^2}=\frac{\partial}{\partial f(x)}\dotprod{\mathcal{A}f}{\mathcal{A}f}=\frac{\partial}{\partial f(x)}\dotprod{f}{(\mathcal{A}^* \mathcal{A})f}\;,
	$$
	where $\mathcal{A}^*(x,y)=\mathcal{A}^*(y,x)$.
	So, for $\mathcal{M}=\mathcal{A}^* \mathcal{A}$, we have
	\begin{align*}
	\frac{\partial}{\partial f(x)}{\normtwo{\mathcal{A}f}^2}&=\frac{\partial}{\partial f(x)}\dotprod{f}{\mathcal{M}f}\\
	&=\sum_{j \in [m]}\int_{I^m_j}\mathcal{M}({I^m_j,i^m(x)})f(y)dy+\sum_{j\in [m]}\int_{I^m_j}\mathcal{M}({i^m(x),I^m_j})f(y)dy\;.
	\end{align*}
	Therefore,
	\begin{align*}
	&\frac{\partial}{\partial f(x)}\frac{\normtwo{\mathcal{A}f}^2}{\normtwo{f}^2}\\
	&=\frac{\normtwo{f}^2\left(\sum_{j\in[m]}\int_{I^m_j}\mathcal{M}({I^m_j,i^m(x)})f(y)dy+\sum_{j\in[m]}\int_{I^m_j}\mathcal{M}({i^n(x),I^m_j})f(y)dy\right)}{\normtwo{f}^4}\\
	&\qquad\qquad\qquad\qquad\qquad\qquad\qquad- \frac{2\normtwo{\mathcal{A}f}^2\cdot f(x)}{\normtwo{f}^4}\;.
	\end{align*}

	%\ynote{$M$ is a dikernel, but $M_{i,i^n(x)}$ looks it's a matrix.}
	%Note that $$M(x,y)=(\mathcal{A}^* \mathcal{A})(x,y)=\int_{0}^{1}\mathcal{A}^*(x,z)\mathcal{A}(z,y) dz=\sum_{i\in [n]}\mathcal{A}(I_i,x)\mathcal{A}(I_i,y) \;,$$
	%Consider a sequence $\{f^i\}_{i \in \mathbb{N}}$ of functions for which the objective values converge to the supremum.
	%Then, the partial derivatives must converge to 0 almost everywhere.
	%By the form of the partial derivatives, for any $\epsilon > 0$, there exists $N \in \mathbb{N}$ such that $|f^N(z)-f^N(z')|<\epsilon$ for almost all $z,z'\in[0,1]$ such that $i^m(z)=i^m(z')$, that is, $f^N$ is almost constant on each of the intervals $I^m_1,\ldots,I^m_m$.
	Consider the optimal solution $f^*$. By the form of the partial derivative, it holds that $f^*(z)=f^*(z')$ for almost all $z,z'\in[0,1]$ such that $i_n(z)=i_n(z')$. That is, $f^*$ is almost constant on each of the intervals $I^m_1,\ldots,I^m_m$.
	Hence, we can define $\bmv \in \R^m$ as $v_j=f^*(x)$, where $x$ is the dominant value in $I^m_j$.
	Then,
	\begin{align*}
	\normtwo{A\bmv}^2&=\sum_{i \in [n]}\left(\sum_{j\in[m]}A_{ij}v_j\right)^2
	=nm^2\int_{0}^{1}\left(\sum_{j\in[m]}\int_{I^m_j}\mathcal{A}(x,I^m_j)f^*(y)dy\right)^2dx \\%- O(\epsilon^2 nm^2 \|A\|_{\max}^2 )\\
	&=nm^2\int_{0}^{1}\left(\int_{0}^{1}\mathcal{A}(x,y)f^*(y)dy\right)^2dx %- O(\epsilon^2 nm^2 \|A\|_{\max}^2 )\\
	=nm^2 \normtwo{\mathcal{A}f^*}^2 %- O(\epsilon^2 nm^2 \|A\|_{\max}^2 )\;,
	\end{align*}
	Moreover,
	$$
	\normtwo{f^*}^2=\int_{0}^{1}f^N(x)^2dx=\sum_{j\in[m]}\int_{I^m_j}f^*(x)^2dx=\frac{1}{m}\sum_{j\in[m]}v_j^2%+O\Bigl(\frac{\epsilon^2}{m}\Bigr)=\frac{1}{m}\normtwo{\bmv}^2+O\Bigl(\frac{\epsilon^2}{m}\Bigr)\;.
	$$
	Therefore, we have $\max_{v\in\R^m} \frac{\normtwo{A\bmv}^2}{\normtwo{\bmv}^2}\ge nm\cdot \sup_{f:[0,1] \to \R}\frac{\normtwo{\mathcal{A}f}^2}{\normtwo{f}^2}$.
\end{proof}

\begin{corollary}\label{cor:NormIdent}
	Let $A\in \R^{n\times m}$ be a matrix.
	Then,
	$$
	\normtwo{A}={\sqrt{nm}}\cdot \normtwo{\mathcal{A}}\;.
	$$
\end{corollary}
\begin{proof}
	The proof is immediate by the definition of the spectral norm and Lemma~\ref{lem:normRelation}. 
	
\end{proof}

\begin{definition}
	Let $\mu$ be a Lebesgue measure.
	A map $\pi\colon [0,1]\to[0,1]$ is \textit{measure preserving} if the pre-image $\pi^{-1}(X)$ is measurable for every measurable set $X$ and $\mu(\pi^{-1}(X))=\mu(X)$.
	A \textit{measure preserving bijection} is a measure preserving map whose inverse map exists and is also measurable.
	For a measure preserving bijection $\pi$ and a dikernel $\mathcal{A}$, we define the dikernel $\pi(\mathcal{A})$, as $\pi(\mathcal{A})(x,y)=\mathcal{A}(\pi(x),\pi(y))$.
\end{definition}

\subsection{Sampling Lemmas}
In this subsection, we will prove that given matrices $A_1,\ldots,A_T\in[-L,L]^{n\times m}$, we obtain a good approximation of their corresponding dikernels, by sampling a small number of elements. The next lemma states that there is a way to ``align'' the sampled matrices with the original matrices. %We refer the reader to Appendix~\ref{subsec:DikernelsProofs} for the full proofs.
\begin{lemma}\label{lem:ApproxBlockSample}
	Given matrices $A_1,\ldots,A_T\in [-L,L]^{n\times n}$ and $\gamma \in (0,1)$, let ${A_1^{\text{str}}},\ldots,{A_T^{\text{str}}}$ be the block approximation matrices as in Lemma~\ref{lem:GoodBlockApprox}. In addition, for $t \in [T]$, let $\mathcal{P}^{A_t}_R=\{V^{A_t}_1,\ldots,V^{A_t}_p\} $ and $\mathcal{P}^{A_t}_C=\{V'^{A_t}_1,\ldots,V'^{A_t}_q\} $ be the row and column partitions of ${A_t^{\text{str}}}$ (from Lemma~\ref{lem:GoodBlockApprox}).
	Let $S_R$ be a set of size $s_R$, generated by picking each element in $[n]$ independently with probability $k_R/n$, and let $S_C$ be a set of size $s_C$, generated by picking each element in $[m]$ independently with probability $k_C/m$ for some $k_R,k_C > 0$.
	
	Then, there exists a measure preserving bijection $\pi\colon[0,1]\to [0,1]$ such that for every $t \in [T]$
	\[
	\exxp{S_R,S_C}\left[\normtwo{\mathcal{A}_t^{\text{str}}-\pi (\mathcal{A}_t^{\text{str}}|_{S_R\times S_C})}\right] =O\left(\frac{L}{\gamma^{11}}\cdot\max\left(\sqrt{\frac{p^{T/2}}{s_R^{1/2}}},\sqrt{\frac{q^{T/2}}{s_C^{1/2}}}\right)\right)\;.
	\]
	
\end{lemma}
\medskip
\begin{proof} 
	We first note that for any $t \in [T]$, any refinement of the row or column partitions of $A_t^{\text{str}}$ will give the same block approximation matrix $A_t^{\text{str}}$.
	Therefore, let $\mathcal{P}_R=\{V^R_1,\ldots,V^R_P\}$ be a partition which refines $\mathcal{P}_R^{A_1},\ldots,\mathcal{P}_R^{A_T}$ and its size is $P=O(p^T)$. Similarly, let $\mathcal{P}_C=\{V^C_1,\ldots,V^C_Q\}$ be a partition which refines $\mathcal{P}_C^{A_1},\ldots,\mathcal{P}_C^{A_T}$ and its size is $Q=O(q^T)$.
	
	We denote by $z^R_i$ the number of elements of $S_R$ that falls into the set $V^R_i$.
	Then, \[\exxp{S_R}[z^R_i]=s_R\cdot\mu(V^R_i)\textrm{ \;and\;} \var[z^R_i]=\mu(V^R_i)(1-\mu(V^R_i))s_R\;.\]
	Similarly, we denote by $z^C_j$ the number of elements of $S_C$ that falls into the set $V^C_j$.
	Then, \[\exxp{S_C}[z^C_j]=s_C\cdot\mu(V^C_j)\textrm{ \;and\;} \var[z^C_j]=\mu(V^C_j)(1-\mu(V^C_j))s_C\;.\]
	
	We next construct a measure preserving bijection.
	We define the following two partitions of the $[0,1]$ interval. 
	
	Let  $\{V^{R'}_1,\ldots,V^{R'}_P\}$ be a partition such that $\mu(V^{R'}_i)=z^R_i/s_R$ and $\mu(V^R_i\cap V^{R'}_i)=\min(\mu(V^R_i),z^R_i/s_R)$, and let  $\{V^{C'}_1,\ldots,V^{C'}_Q\}$ be a partition such that $\mu(V^{C'}_j)=z^C_j/s_C$ and $\mu(V^C_j\cap V^{C'}_j)=\min(\mu(V^C_j),z^C_j/s_C)$.
	We construct the dikernel $\mathcal{Y}:[0,1]^2\to \R$ such that the value of $\mathcal{Y}$ on $V^{R'}_i\times V^{C'}_j$ is the same as the value of $\mathcal{A}^{\text{str}}$ on $V^R_i\times V^C_j$. Therefore, the dikernel $\mathcal{Y}$ agrees with $\mathcal{A}^{\text{str}}$ on the set $Y=\bigcup_{(i,j)\in[P]\times[Q]}(V^R_i\cap V^{R'}_i)\times(V^C_j\cap V^{C'}_j)$. Then, there exists a bijection $\pi$ such that $\pi(\mathcal{A}^{\text{str}}|_{S_R\times S_C})=\mathcal{Y}$. Then,
	\begin{align*}
	1-\mu(Y)&\le 1-\left(\sum_{i\in [P]}\min\left(\mu(V^R_i),\frac{z^R_i}{s_R}\right)\right)\left(\sum_{j\in [Q]}\min\left(\mu(V^C_j),\frac{z^C_j}{s_C}\right)\right)\\
	&=1-\left(1-\frac{1}{2}\sum_{i\in[P]}\left|\mu(V^R_i)-\frac{z^R_i}{s_R}\right|\right)\left(1-\frac{1}{2}\sum_{j\in[Q]}\left|\mu(V^C_j)-\frac{z^C_j}{s_C}\right|\right)\\
	&\le \max\left\{\left(P\sum_{i\in[P]}\left(\mu(V^R_i)-\frac{z^R_i}{s_R}\right)^{2}\right)^{1/2},\left(Q\sum_{j\in[Q]}\left(\mu(V^C_j)-\frac{z^C_j}{s_C}\right)^{2}\right)^{1/2} \right\}\;.
	\end{align*}
	Therefore, we have that $$(1-\mu(Y))^2\le \max\left\{P\sum_{i\in[P]}\left(\mu(V^R_i)-\frac{z^R_i}{s_R}\right)^2, Q\sum_{j\in[Q]}\left(\mu(V^C_j)-\frac{z^C_j}{s_C}\right)^2\right\}\;.$$
	Taking expectation (over the choice of $S_R$ and $S_C$) yields,
	\[
	\exxp{S_R,S_C}\left[1-\mu(Y)\right]\le \max \left(\sqrt{\frac{Q}{s_C}},\sqrt{\frac{P}{s_R}}\right)\;.
	\]
	
	% In Particular, we have that for all $i\in [p]$ the expected symmetric difference, \[\exxp{S}\left[\mu(V_i\triangle V_i')\right]=\sqrt{\frac{p}{s}}\;.\]

	Let  $\mathcal{U}=\mathcal{A^{\text{str}}}-\mathcal{Y}$ and consider a corresponding matrix $U$. $U$ is an $N\times M$ matrix, where $N=\textrm{lcm}(n\cdot\mu(V^R_1\triangle V^{R'}_1),\ldots,n\cdot\mu(V^R_P\triangle V^{R'}_P)) $ and $M=\textrm{lcm}(m\cdot\mu(V^C_1\triangle V^{C'}_1),\ldots,m\cdot\mu(V^C_Q\triangle V^{C'}_Q)) $. %\ynote{$\mu(\cdot)$ is at most one and the lcm is not well defined. I guess you want to multiply them by $n$.}
	By Claim~\ref{lem:BoundedValuesSTR}, the absolute value of an entry in the matrix $A^{\text{str}}$ is bounded by $\frac{2}{\gamma^{11}}L$, and thus the absolute value of an entry in $U$ is bounded by $\frac{4}{\gamma^{11}}L$.
	
	Then,
	\begin{align*}
	\exxp{S_R,S_C}[\normtwo{U}^2]&\le \exxp{S_R,S_C}[\normF{U}^2]=\exxp{S_R,S_C}\left[{\sum_{i=1}^{N}\sum_{j=1}^{M}U_{ij}^2}\right]\\&\le \exxp{S_R,S_C}\left[{NM(1-\mu(Y))\cdot\max_{(i,j) \in [N]\times[M]}(U)_{ij}^2}\right]\\
	&\le NM\cdot\max_{(i,j) \in [N]\times[M]}U_{ij}^2\cdot\exxp{S_R,S_C}\left[1-\mu(Y)\right] \\
	&\le NM\cdot  \left(\frac{4}{\gamma^{11}}\right)^2L^2\cdot\max \left(\sqrt{\frac{Q}{s_C}},\sqrt{\frac{P}{s_R}}\right) \;.
	\end{align*}
	
	Using Corollary~\ref{cor:NormIdent} we get $\exxp{S_R,S_C}\left[\normtwo{\mathcal{U}}\right]\le \frac{4L}{\gamma^{11}}\max\left(\sqrt{\frac{p^{T/2}}{s_R^{1/2}}},\sqrt{\frac{q^{T/2}}{s_C^{1/2}}}\right)$, and the lemma follows.
\end{proof}

\medskip
\noindent In the following lemma we prove concentration around the mean.
\begin{lemma}\label{lem:BoundTwonorm}
	Let $\gamma>0$, and ${A_1,\ldots,A_T}\in [-L,L]^{n\times m}$.
	Let $S_R$ be a set generated by picking each element in $[n]$ independently with probability $k_R/n$ and let $S_C$ be a set generated by picking each element in $[m]$ independently with probability $k_C/m$ for some $k_R,k_C > 0$.
	Then, with probability at least $89/100$ there exists a measure preserving bijection $\pi\colon[0,1]\to[0,1]$ such that for every $t \in [T]$, \begin{align*}
	\normtwo{\mathcal{A}_t-\pi\left(\mathcal{A}_t|_{S_\text{R}\times S_\text{C}}\right)}&\le 210\gamma L T+10L T\left(\sqrt{\frac{8\log n}{k_C}}+\sqrt{\frac{8\log m}{k_R}}+\sqrt{\frac{4\log n\log m}{k_Rk_C}}\right)\\
	&+ O\left(\frac{LT}{\gamma^{11}}\left(\frac{1}{\gamma^{10}}\right)^{\frac{3T}{4\gamma^2}}\cdot\max\left(\frac{1}{k_R^{1/4}},\frac{1}{k_C^{1/4}}\right)\right),
	\end{align*}
	where $N = \sqrt{nm}$.
\end{lemma}

In order to prove the lemma we introduce the following result regarding a random submatrix from~\cite{Tropp:2015vy} section 5.2.2. %\ynote{kind to mention the number of the lemma because the survey is long.}
\begin{lemma}\label{lem:RandMatrixNorm}
	Given a matrix $A\in [-L,L]^{n\times m}$, let $P={\diag}(\chi_1,\ldots,\chi_n)$ be the diagonal matrix where $\{\chi_i\}$'s are $\textrm{Bernoulli}(k_R/n)$ random variables for $k_R>0$.
	In addition, let $R={\diag}(\xi_1,\ldots,\xi_m)$ be the diagonal matrix where $\{\xi_j\}$'s are $\textrm{Bernoulli}(k_C/m)$ random variables for $k_C > 0$.
	Then,
	\[\exxp{}[\normtwo{PAR}^2]\le\frac{3k_R\cdot k_C}{nm}\normtwo{A}^2 +2k_RL^2\log n+2k_CL^2\log m+L^2\log n\log m \;.\]
\end{lemma}
The above lemma shows that a random submatrix of size roughly $k_R\times k_C$ gets its ``fair share'' of the spectral norm of $A$.

\begin{proofof}{Lemma~\ref{lem:BoundTwonorm}}
	Since $S_R$ and $S_C$ are set of indices generated by choosing each index to $S_R$ (or $S_C$) with probability $k_R/n$ ($k_C/m$), we have that with probability at least $ 99/100$, \[|S_R|\ge k_R/2\;\;\text{and }|S_C|\ge k_C/2\;.\]We henceforth condition on that.
	
	For any measure preserving bijection $\pi:[0,1]\to[0,1]$, and $t \in [T]$ we have
	\begin{align*}
	\exxp{S_R,S_C}[\normtwo{\mathcal{A}_t-\pi(\mathcal{A}_t|_{S_\text{R}\times S_\text{C}})}] &\le \normtwo{\mathcal{A}_t-\mathcal{A}_t^{\text{str}}}+\exxp{S_R,S_C}\left[\normtwo{\mathcal{A}_t^{\text{str}}-\pi (\mathcal{A}_t^{\text{str}}|_{S_\text{R}\times S_\text{C}})}\right]\\
	&+\exxp{S_R,S_C}\left[\normtwo{\pi(\mathcal{A}_t^{\text{str}}|_{S_\text{R}\times S_\text{C}})-\pi (\mathcal{A}_t|_{S_\text{R}\times S_\text{C}})}\right]\;.
	\end{align*}
	
	Where $A_t^{\text{str}}$ is the matrix obtained by Lemma~\ref{lem:GoodBlockApprox}.

	\noindent By Lemma~\ref{lem:GoodBlockApprox}  and Corollary~\ref{cor:NormIdent}, we have that for any $t \in [T]$ $$\normtwo{\mathcal{A}_t-\mathcal{A}_t^{\text{str}}}\le 7\gamma L\;.$$

	\noindent By the facts that $|S_R|\ge k_R/2$, $|S_C|\ge k_C/2$,  $p=O\left(\left(\frac{1}{\gamma^{10}}\right)^{3/\gamma^2}\right)$ and $q=O\left(\left(\frac{1}{\gamma^{10}}\right)^{3/\gamma^2}\right)$, Lemma~\ref{lem:ApproxBlockSample} yields that for any $t \in [T]$:
	\[\exxp{S_R,S_C}\left[\normtwo{\mathcal{A}_t^{\text{str}}-\pi (\mathcal{A}_t^{\text{str}}|_{S_R\times S_C})}\right]=O\left(\frac{L}{\gamma^{11}}\left(\frac{1}{\gamma^{10}}\right)^{\frac{3T}{4\gamma^2}}\cdot\max\left(\frac{1}{k_R^{1/4}},\frac{1}{k_C^{1/4}}\right)\right)
	\]

	We are left with bounding $\exxp{S_R,S_C}\left[\normtwo{\pi(\mathcal{A}_t^{\text{str}}|_{S_R\times S_C})-\pi (\mathcal{A}_t|_{S_R\times S_C})}\right]$.
	We apply Lemma~\ref{lem:RandMatrixNorm}  on $(A_t^{\text{str}}-A_t)|_{S_R\times S_C}$  to get
	\begin{align*}
	&\exxp{S_R,S_C}\left[\normtwo{(A_t^{\text{str}}-A_t)|_{S_R\times S_C}}^2\right]\\
	&\qquad\qquad\le \frac{3k_C\cdot k_R}{nm}\normtwo{A^{\text{str}}-A}^2+2L^2k_R\log n+2L^2k_C\log m +L^2\log n\log m \\
	&\qquad\qquad \le \frac{3k_C\cdot k_R}{nm} \cdot 16\gamma^2L^2nm+2L^2k_R\log n+2L^2k_C\log m +L^2\log n\log m\\
	&\qquad\qquad \le 48L^2\gamma^2k_R\cdot k_C +2L^2k_R\log n+2L^2k_C\log m +L^2\log n\log m\;,
	\end{align*}
	which implies that,
	\begin{align*}
	&\exxp{S_R,S_C}\left[\normtwo{(A_t^{\text{str}}-A_t)|_{S_R\times S_C}}\right]\\
	&\qquad\qquad\le 7\gamma L\sqrt{k_C\cdot k_R}+L\sqrt{2k_R\log n }+L\sqrt{2k_C\log m }+L\sqrt{\log m\log n} \;.\end{align*}
	By applying Corollary~\ref{cor:NormIdent} and using the fact that the dimension of $(A^{\text{str}}-A)|_{S_R\times S_C}$ at least $k_R\cdot k_C/4$, we get
	\[\exxp{S_R,S_C}\left[\normtwo{\mathcal{A}_t^{\text{str}}|_{S_R\times S_C}-\mathcal{A}_t|_{S_R\times S_C}}\right]\le14\gamma L+L\sqrt{\frac{8\log n}{k_c} }+L\sqrt{\frac{8\log m}{k_R} }+L\sqrt{\frac{4\log m\log n}{k_R\cdot k_C}}\;.\]
	Putting everything together,
	\begin{align*}
	\exxp{S_R,S_C}[\normtwo{\mathcal{A}_t-\pi(\mathcal{A}_t|_{S_R\times S_C})}] &\le 21\gamma L+L\left(\sqrt{\frac{8\log n}{k_C}}+\sqrt{\frac{8\log m}{k_R}}+\sqrt{\frac{4\log n\log m}{k_Rk_C}}\right)\\
	&+ O\left(\frac{L}{\gamma^{11}}\left(\frac{1}{\gamma^{10}}\right)^{\frac{3T}{4\gamma^2}}\cdot\max\left(\frac{1}{k_R^{1/4}},\frac{1}{k_C^{1/4}}\right)\right)\;.
	\end{align*}
	By Markov inequality, with probability at least $9/10T$,
	\begin{align*}
	\normtwo{\mathcal{A}_t-\pi(\mathcal{A}_t|_{S_\text{R}\times S_\text{C}})}&\le 210\gamma LT+10LT\left(\sqrt{\frac{8\log n}{k_C}}+\sqrt{\frac{8\log m}{k_R}}+\sqrt{\frac{4\log n\log m}{k_Rk_C}}\right)\\
	&+ O\left(\frac{LT}{\gamma^{11}}\left(\frac{1}{\gamma^{10 }}\right)^{\frac{3T}{4\gamma^2}}\cdot\max\left(\frac{1}{k_R^{1/4}},\frac{1}{k_C^{1/4}}\right)\right)\;.
	\end{align*}
	By using a union bound the lemma follows.
\end{proofof}

\section{Applications}\label{sec:quadratic}
\subsection{ Quadratic Function Minimization}
In this section, we show that we can approximately solve quadratic function minimization problems in polylogarithmic time.

Recall that we are given a matrix $A\in \R^{n\times n}$ and vectors $\bmd,\bmb\in \R^{n}$, and consider the following quadratic function minimization problem:
\begin{align}
\qquad\min_{\bmv\in \R^n} \psi_{n,A,\bmd,\bmb}(\bmv),\; \text{where } \psi_{n,A,\bmd,\bmb}(\bmv)= \dotprod{\bmv}{A\bmv}+n\dotprod{\bmv}{\diag(\bmd)\bmv}+n\dotprod{\bmb}{\bmv}\;.\label{Eq:MinDiscreat}
\end{align}
Here $\diag(\bmd) \in \R^{n \times n}$ is a matrix whose diagonal entries are specified by $\bmd$.

First, we describe our algorithm for minimizing quadratic functions.
We first sample a set of indices $S \subseteq \set{1,2,\ldots,n}$ with each index included with probability $k/n$, where $k$ is some constant.
If $|S|$ is too large, we immediately stop the process by claiming that the algorithm has failed.
Otherwise, we solve the problem on $A|_S$, $\bmd|_S$, $\bmb|_S$ and then output the optimal solution.
The detail is given in Algorithm~\ref{alg:MainAlg}.

\begin{algorithm}[h!]
	\caption{Minimization Algorithm($A$, $n$, $\eps$, $k$) }\label{alg:MainAlg}
	\begin{algorithmic}[1]
		\State Let $S\subseteq \set{1,2,\ldots,n}$ such that each index $i$ is taken to $S$ independently w.p $k/n$.\label{alg:Stage1}
		\If {$|S|>2k$}
		\State \textbf{Abort}
		\EndIf
		\State \Return $\min_{\bmv\in \R^{|S|}}\psi_{|S|,A|_S,\bmd|_S,\bmb|_S}(\bmv)$.
	\end{algorithmic}
\end{algorithm}

Due to our extensive use of dikernels in the analysis, we introduce a continuous version of problem (\ref{Eq:MinDiscreat}). The real valued function $\Psi_{n,A,\bmd,\bm b}$ on function $f\colon [0,1]\to \R$ is defined as
$$\Psi_{n,A,\bmd,\bmb}(f)=\dotprod{f}{\mathcal{A}f}+\dotprod{f^2}{\mathcal{D}{\mathbf{1}}}+\dotprod{f}{\mathcal{B}{\mathbf{1}}}\;,$$
where $\mathcal{D}$ and $\mathcal{B}$ are the corresponding dikernels of $\bmd\cdot \bmone^\top$ and $\bmb\cdot \bmone^\top$ respectively, $f^2\colon[0,1]\to\R$ is a function such that $f^2(x)=f(x)^2$ for every $x\in[0,1]$ and $\bmone\colon[0,1]\to\R$ is the constant  function that has the value $1$ everywhere.

In order to prove Theorem~\ref{thm:mainTheorem-intro}, we prove that the minimizations of $\psi_{n,A,\bmd,\bmb}$ and $\Psi_{n,A,\bmd,\bmb}$ are essentially equivalent. 
\begin{lemma}\label{lem:Min-Inf}
	Let $A\in \R^{n\times n}$ and $\bmb,\bmd\in \R^{n}$. Then, for any $r>0$
	$$\min_{\bmv:\normtwo{\bmv}\le r}\psi_{n,A,\bmd,\bmb}(\bmv)=n^2\cdot \inf_{f:\normtwo{f}\le \frac{r}{\sqrt{n}}}\Psi_{n,A,\bmd,\bmb}(f) \;.$$
\end{lemma}
\begin{proof} 
	In contrast to the proof of Lemma~\ref{lem:normRelation}, in this case we have to deal with constrained optimization, and therefore must consider the KKT optimality conditions. We start by showing that $\min_{\bmv:\normtwo{\bmv}\le r}\psi_{n,A,\bmd,\bmb}(\bmv)\ge n^2\cdot \inf_{f:\normtwo{f}\le \frac{r}{\sqrt{n}}}\Psi_{n,A,\bmd,\bmb}(f)$. Given a vector $\bmv \in \R^n$ such that $\normtwo{\bmv}\le r$, we define the function $f:[0,1]\to\R$ as $f(x)=v_{i^n(x)}$.
	Then,
	\begin{align*}
	&\dotprod{f}{\mathcal{A}f}=\sum_{i,j\in[n]}\int_{I^n_i}\int_{I^n_j}A_{ij}f(x)f(y)dxdy=\frac{1}{n^2}\dotprod{\bmv}{A\bmv}\\
	&\dotprod{f^2}{\mathcal{D}\bm{1}}=\sum_{i,j\in[n]}\int_{I^n_i}\int_{I^n_j}d_i f(x)^2dxdy=\frac{1}{n}\dotprod{\bmv}{\diag(\bm{d})\bmv}\\
	&\dotprod{f}{\mathcal{B}\bm{1}}=\sum_{i,j\in[n]}\int_{I^n_i}\int_{I^n_j}b_i f(x)dxdy=\frac{1}{n}\dotprod{\bmv}{\bmb}
	\end{align*}
	%\ynote{Define $\hat{A}(x,I_j)$}
	%  \amargin{Added definition of $\hat{A}(x,I_j)$ }
	In addition,
	$$
	\normtwo{f}^2=\int_{0}^{1}f(x)^2dx=\sum_{j\in[n]}\int_{I^n_j}f(x)^2dx=\frac{1}{n}\sum_{j\in[n]}v_j^2=\frac{1}{n}\normtwo{\bmv}^2\le \frac{r^2}{n}\;.
	$$
	
	Next, we show that $\min_{\bmv:\normtwo{\bmv}\le r}\psi_{n,A,\bmd,\bmb}(\bmv)\le n^2\cdot \inf_{f:\normtwo{f}\le \frac{r}{\sqrt{n}}}\Psi_{n,A,\bmd,\bmb}(f)$.
	First, we note that the latter problem has a minimizer $f:[0,1]\to \R $ because it is weakly continuous and coercive (See, e.g.,~\cite{Peressini:1993ug}).
	According to the generalized KKT conditions (see, e.g., Section~9.4 of~\cite{Luenberger:1997:OVS:524110}), there exists $\lambda$ such that:
	\begin{itemize}
		\item (Stationarity) $\frac{\partial}{\partial f^*(x)}\Psi_{n,A,\bmd,\bmb}(f^*(x))- \lambda \frac{\partial}{\partial f^*(x)}(\normtwo{f^*}-r/\sqrt{n})=0$ for almost all $x$.
		\item (Primal feasibility) $\normtwo{f^*}-r/\sqrt{n}\le 0$
		\item (Complementary slackness) $\lambda\cdot (\normtwo{f^*}-r/\sqrt{n})=0$
	\end{itemize}
	
	The stationarity condition yields:
	\begin{align*}
	&\frac{\partial}{\partial f^*(x)}\Psi_{n,A,\bmd,\bmb}(f^*(x))- \lambda \frac{\partial}{\partial f^*(x)}(\normtwo{f^*}-r/\sqrt{n})\\
	&= \sum_{i\in {n}}\int_{I^n_i}A_{ii^n(x)}f^*(y)dy+\sum_{j\in {n}}\int_{I^n_j}A_{i^n(x) j}f^*(y)dy+2d_{i^n(x)}f^*(x) +b_{i^n(x)}-2\lambda f^*(x) \;,
	\end{align*}
	By the form of the partial derivatives, $f^*(z)=f^*(z')$ for almost all $z,z'\in[0,1]$ such that $i^n(z)=i^n(z')$.
	That is, $f^*$ is almost constant on each of the intervals $I^n_1,\ldots,I^n_n$. Therefore, we define $\bmv\in \R ^n$ as $v_j=f^*(x)$, where $x$ is the dominant value in $I^n_j$. Then,
	\begin{align*}
	&\dotprod{\bmv}{A\bmv}=\sum_{i,j\in [n]}A_{ij}v_iv_j = n^2\sum_{i,j\in [n]}\int_{I^n_i}\int_{I^n_j}A_{ij}f^*(x)f^*(y)dxdy  %- O\left(\eps n^2\|A\|_{\text{max}}\right)\\
	=n^2\dotprod{f^*}{\mathcal{A}f^*}\;. \\ %-O\left(\eps n^2\|A\|_{\text{max}}\right)\;.\\
	&\dotprod{\bmv}{\diag(\bm{d})\bmv}=\sum_{i\in[n]}d_i v_i^2= n\sum_{i\in[n]}\int_{I^n_i}d_i f^*(x)^2dx %- O\left(\eps^2 n \|\diag(\bmd)\|_{\text{max}}\right)\\
	=n\dotprod{(f^*)^2}{\mathcal{D}\bm{1}}\;. \\ %-O\left(\eps^2 n \|\diag(\bmd)\|_{\text{max}}\right).\\
	&\dotprod{\bmv}{\bmb}=\sum_{i\in[n]}b_iv_i
	=n\sum_{i\in[n]}\int_{I^n_i}b_i f^*(x)dx %-O\left(\eps n\|\bmb\|_{\text{max}}\right)
	= n\dotprod{f^*}{\mathcal{B}\bm{1}}. %-O\left(\eps n\|\bmb\|_{\text{max}}\right).
	\end{align*}
	
	In addition,
	$
	\normtwo{f^*}^2=\int_{0}^{1}f^*(x)^2dx=\sum_{j\in[n]}\int_{I^n_j}f^*(x)^2dx
	=\frac{1}{n}\normtwo{\bmv}^2 %+O\Bigl(\frac{\epsilon^2}{n}\Bigr)
	\le\frac{r^2}{n}\;. %+O\Bigl(\frac{\epsilon^2}{n}\Bigr)\;.
	$ 
	
	Hence, we get that $$\min_{\bmv:\normtwo{\bmv}\le r}\psi_{n,A,\bmd,\bmb}(\bmv)\le n^2\cdot \inf_{f:\normtwo{f}\le \frac{r}{\sqrt{n}}}\Psi_{n,A,\bmd,\bmb}(f)\;,$$
	and the lemma follows.
\end{proof}

With the above lemma, we are ready to prove our main result.

\begin{proofof}{Theorem~\ref{thm:mainTheorem-intro}}
	By applying Chernoff bounds, we have that with probability at least $1-o(1)$, the size of $S$ is at most $2k$.
	As before, we apply Lemma~\ref{lem:BoundTwonorm} with $\gamma=O(\eps)$,  $$k=\max\Bigg\{  O\left(\frac{\log^2 n}{\eps^2}\right), \left(\frac{1}{\eps}\right)^{O(1/\eps^2)}  \Bigg\}\;,$$ and $S_C=S_R=S$\footnote{We note that in this case, the sets $S_R$ and $S_C$ are dependent. However, the proof of Lemma~\ref{lem:BoundTwonorm} can be easily modified to the case where the matrix is in $[-L,L]^{n \times n}$, where for this case $S_R=S_C=S$.}.
	Then, with probability at least $2/3$ there exists a measure preserving bijection $\pi\colon[0,1]\to [0,1]$ such that for any function $f\colon[0,1]\to \R$,
	\[
	\max \Big\{\Big|\dotprod{f}{(\mathcal{A}-\pi(\mathcal{A|_S}))f}\Big|,\Big|\dotprod{f}{(\mathcal{D}-\pi(\mathcal{D|_S}))f}\Big|,\Big|\dotprod{f}{(\mathcal{B}-\pi(\mathcal{B|_S}))f}\Big|\Big\}\le \frac{\eps L}{3} \normtwo{f}^2\;.
	\]
	Let $R=\max\left\{\frac{\normtwo{\tilde{\bmv}^*}}{\sqrt{|S|}},\frac{\normtwo{\bmv^*}}{\sqrt{n}}\right\}$. Then, by using Lemma~\ref{lem:Min-Inf}:
	\begin{align*}
	\tilde{z}^*&=\min_{\bmv\in \R^{|S|}}\psi_{|S|,A|_S,\bmd|_S,\bmb|_S}(\bmv)=\min_{\bmv:\normtwo{\bmv}\le R\sqrt{|S|}}\psi_{|S|,A|_S,\bmd|_S,\bmb|_S}(\bmv)\\
	&=|S|^2\cdot\min_{f:\normtwo{f}\le R}\Psi_{|S|,A|_S,\bmd|_S,\bmb|_S}(f)\\
	&=|S|^2\cdot \min_{f:\normtwo{f}\le R}\Big\{\dotprod{f}{\left(\pi(\mathcal{A|_S})-\mathcal{A}\right)f}+\dotprod{f}{\mathcal{A}f}+\dotprod{f^2}{\left(\pi(\mathcal{D|_S})-\mathcal{D}\right)\mathbf{1}}+\dotprod{f^2}{\mathcal{D}\mathbf{1}}\\
	&\;\;\;\;\;\;\;\;\;\;\;\;\;\;\;\;\;\;\;\;\;\;\;\;\;+\dotprod{f}{\left(\pi(\mathcal{B|_S})-\mathcal{B}\right)\mathbf{1}}+\dotprod{f}{\mathcal{B}\mathbf{1}}\Big\} \\
	&\le |S|^2\cdot \min_{f:\normtwo{f}\le R}\left\{\dotprod{f}{\mathcal{A}f}+\dotprod{f^2}{\mathcal{D}\mathbf{1}}+\dotprod{f}{\mathcal{B}\mathbf{1}} \pm\eps L\normtwo{f}^2\right\}\\
	&\le |S|^2\cdot \min_{f:\normtwo{f}\le R}\Psi_{n,A,\bmd,\bmb}(f)\pm\eps  L|S|^2R^2\\
	&= \frac{|S|^2}{n^2}\cdot \min_{\bmv:\normtwo{\bmv}\le\sqrt{n}R}\psi_{n,A,\bmd,\bmb}(\bmv)\pm\eps L|S|^2R^2\\
	&= \frac{|S|^2}{n^2}\cdot \min_{\bmv\in \R^n}\psi_{n,A,\bmd,\bmb}(\bmv)\pm\eps L|S|^2 R^2
	= \frac{|S|^2}{n^2}z^*\pm\eps L |S|^2 R^2\;.
	\end{align*}
	By rearranging the inequality and applying the union bound the theorem follows.
\end{proofof}

As a corollary we show that we can obtain an approximation algorithm for minimizing a quadratic function over a ball of radius $r$ with better error bounds compared than the one obtained by Hayashi and Yoshida \cite{Hayashi:2016wh}. The proof of correctness is similar to the proof of Theorem~\ref{thm:mainTheorem-intro}.

\begin{algorithm}[ht!]
	\caption{Minimization Algorithm Over a Ball($A$, $n$, $\eps$, $k$, $r$) }\label{alg:MainAlgBall}
	\begin{algorithmic}[1]
		\State Let $S\subseteq \set{1,2,\ldots,n}$ such that each index $i$ is taken to $S$ independently w.p $k/n$.\label{alg:Stage1Ball}
		\If {$|S|>2k$}
		\State \textbf{Abort}
		\EndIf
		\State \Return $\min_{\normtwo{v}\le\sqrt{\frac{|S|}{n}}r}\psi_{|S|,A|_S,\bmd|_S,\bmb|_S}(\bmv)$.
	\end{algorithmic}
\end{algorithm}

\begin{corollary}[Restatement of Theorem~\ref{Eq:MinDiscreatConstrained-intro}] Let $\bmv^*$ and $z^*$ be an optimal solution and optimal value, respectively, of $\min_{\normtwo{\bmv}\le r} \psi_{n,A,\bmd,\bmb}(\bmv)$. Let $\eps>0$ and let $S$ be a random set generated as in Algorithm~\ref{alg:MainAlgBall}  with
	$$k=\max\Bigg\{  O\left(\frac{\log^2 n}{\eps^2}\right), \left(\frac{1}{\eps}\right)^{O(1/\eps^2)}  \Bigg\}\;.$$  Then, we have that with probability at least $2/3$, the  following hold: Let $\tilde{\bmv}^*$ and $\tilde{z}^*$ be an optimal solution and the optimal value, respectively, of the problem $$\min_{\normtwo{v}\le\sqrt{\frac{|S|}{n}}r}\psi_{|S|,A|_S,d|_S,b|_S}(\bmv)\;.$$ Then,
	$$\left|\frac{1}{|S|^2}\tilde{z}^*-\frac{1}{n^2}{z^*}\right|\le \frac{\eps L r^2}{n}\;,$$
	where $L=\max\{\max_{i,j}|A_{ij}|,\max_i |d_i|,\max_i|b_i| \}$.
\end{corollary}
\subsection{ Singular Values Approximation Algorithm}\label{sec:Eigen}
As an additional application for our method, we show that we can obtain an approximation algorithms for the top singular values of a given matrix. We note that similar results (with better running time) can be obtained by applying known sampling techniques from~\cite{frieze2004fast}.
In Subsection~\ref{ssec:LargestEigenvalue}, as a warm-up, we will prove the correctness of Algorithm~\ref{alg:EigenAlg}, which approximate the largest singular value.
The algorithm is simple, and will allow us to demonstrate the use of our method.
In addition, Algorithm~\ref{alg:EigenAlg} has better error guarantee.
In Subsection~\ref{ssec:TopEigen}, we generalize the ideas behind Algrorithm~\ref{alg:EigenAlg}, and prove the result for the top singular values.
\subsubsection{Warmup- Approximating the Largest Singular Value}\label{ssec:LargestEigenvalue}
In this subsection we analyze the following algorithm.
\begin{algorithm}[ht!]
	\caption{Approximate the largest Singular Value ($A$, $n$, $m$, $\eps$) }\label{alg:EigenAlg}
	\begin{algorithmic}[1]
		\State Let $k=\max\Bigg\{  O\left(\frac{\max\{\log^2 n,\log^2 m\}}{\eps^2}\right), \left(\frac{1}{\eps}\right)^{O(1/\eps^2)}  \Bigg\}$
		\State Let $S_R\subseteq \set{1,2,\ldots,n}$ such that each index $i$ is taken to $S_R$ independently with probability $k/n$.\label{alg:eigenStage1}
		\State Let $S_C \subseteq \set{1,2,\ldots,m}$ such that each index $i$ is taken to $S_C$ independently with probability $k/m$.\label{alg:eigenStage2}
		\If {$|S_R|>2k$ or $|S_C|>2k$}
		\State \textbf{Abort}
		\EndIf
		\State \Return $\sqrt{\frac{nm}{|S_R||S_C|}}\sigma_1(A|_{S_R\times S_C})$.
	\end{algorithmic}
\end{algorithm}

% We first state a lemma showing that the minimizations of $\dotprod{\bmv}{A\bmv}$ and its continuous version $\dotprod{f}{\mathcal{A}f}$, are equivalent up to normalization.
% \begin{lemma}\label{lem:Min-Inf_special case}
% 	Let $A\in \R^{n\times n}$.  Then, for any $L>0$
% 	\[ \max_{\bmv:\normtwo{\bmv}\le L}\dotprod{\bmv}{A\bmv}=n^2\cdot \sup_{f:\normtwo{f}\le \frac{L}{\sqrt{n}}}\dotprod{f}{\mathcal{A}f} \;.\]
% \end{lemma}
% \begin{proof}
% 	We refer the reader to proof of Lemma~\ref{lem:normRelation} which is similar .
% \end{proof}
%We turn to showing that Algorithm $\ref{alg:EigenAlg}$ well approximate the value of the $\lambda_1(A)$.
%Formally, we prove the following:
\begin{theorem}\label{thm:eigenvalueapprox}
	Given a matrix $A\in {[-L,L]}^{n\times m}$ and $\eps>0$, Algorithm~\ref{alg:EigenAlg} outputs a value $z$ such that with probability at least $2/3$ \[
	|\sigma_1(A)-z|\le \eps L \sqrt{nm}\;.\]
\end{theorem}

\medskip

\begin{proof}%{Theorem~\ref{thm:eigenvalueapprox}}
	By usual Chernoff bound we have that with probability at least $1-o(1)$, the sizes of $S_R$ and $S_C$ are at most $2k$.
	We henceforth condition on that.
	We apply Lemma~\ref{lem:BoundTwonorm} with $\gamma=O(\eps)$ and
	\begin{align*}
	k=k_R=k_C&=
	\max\Bigg\{\frac{\max\{\log^2n,\log^2 m\}}{\gamma^2},\frac{1}{\gamma^{48}}\left(\frac{1}{\gamma^{10}}\right)^{\frac{3}{\gamma^2}}\Bigg\}\\
	&=\max\Bigg\{  O\left(\frac{\max\{\log^2 n,\log^2 m\}}{\eps^2}\right), \left(\frac{1}{\eps}\right)^{O(1/\eps^2)}  \Bigg\}\;.
	\end{align*}
	Then, with probability at least $2/3$, there exists a measure preserving bijection $\pi:[0,1]\to [0,1]$ such that
	\[
	\normtwo{(\mathcal{A}-\pi(\mathcal{A}|_{S_R \times S_C}))} \le \eps L\;.
	\]
	Then, by the definition of $\sigma_1(A)$,
	\begin{align*}
	\sigma_1(A) &= \normtwo{A}= \sqrt{nm} \normtwo{\mathcal{A}} \tag{By Corollary~\ref{cor:NormIdent}}\ \\
	& \leq  \sqrt{nm} \left(\normtwo{\mathcal{A}|_{S_R \times S_C}}+ \normtwo{\pi(\mathcal{A}|_{S_R \times S_C})-\mathcal{A}} \right)\\
	& = \sqrt{nm}\left(\normtwo{\mathcal{A}|_{S_R \times S_C}}\pm \eps L \right)
	= \sqrt{nm}\normtwo{\mathcal{A}|_{S_R \times S_C}}\pm \eps L \sqrt{nm} \\
	&= \sqrt{\frac{nm}{|S_R||S_C|}}\normtwo{A|_{S_R \times S_C}} \pm \eps L \sqrt{nm}\tag{By Corollary~\ref{cor:NormIdent}}\ \\
	&= \sqrt{\frac{nm}{|S_R||S_C|}}\sigma_1(A|_{S_R \times S_C})\pm \eps L \sqrt{nm}.
	\end{align*}
	By applying a union bound, the theorem follows.
\end{proof}

\subsubsection{Approximating the t-th Singular Value}\label{ssec:TopEigen}
In this subsection, we will generalize Algorithm~\ref{alg:EigenAlg} to approximate the $t$-th largest singular value. In order to do so, we consider the following (well known) result regarding best rank-$t$ approximation of a matrix.
\begin{lemma}\label{lem:BestRankApprox}
	Given $A\in\R^{n\times m}$ with left and right singular vectors $\bmu^1,\ldots,\bmu^{\min\set{n,m}}$ and $\bmv^1,\ldots,\bmv^{\min\set{n,m}}$ corresponding to singular values $\sigma_1(A)\ge \ldots\ge \sigma_{\min\set{n,m}}(A)$. For $t \leq \min{\set{n,m}}$, let $A_t=\sum_{\ell\in[t]}\sigma_\ell (A)\bmu^\ell(\bmv^\ell)^\top$.
	Then,  $$\Lambda_t\eqdef\min_{\substack{\hat{\bmu}^1,\ldots,\hat{\bmu}^t\in\R^n \\ \hat{\bmv}^1,\ldots,\hat{\bmv}^t \in \R^m}}\norm{A-\sum_{\ell\in[t]}\hat{\bmu}^\ell(\hat{\bmv}^\ell)^\top}_F^2=\normF{A-A_t}^2=\sum_{\ell=t+1}^{n}\sigma_\ell(A)^2\;.$$
\end{lemma}
The above suggests that if we were able to get an approximation $\widetilde{\Lambda}_t$ and $\widetilde{\Lambda}_{t-1}$, of ${\Lambda}_t$ and ${\Lambda}_{t-1}$, within an additive error of $\beta^2 nm$ (for some $\beta\in(0,1)$), then we could just compute $\sqrt{\widetilde{\Lambda}_t-\widetilde{\Lambda}_{t-1}}$, to get an approximation of $\sigma_t(A)$ within an additive error of $O({\beta}\sqrt{nm})$.

In order to get such approximations, we use the framework used for the analysis of Algorithm \ref{alg:EigenAlg}.
Specifically, we sample a set of indices $S_R\subseteq[n]$ and $S_C\subseteq [m]$, each picked independently with equal probability.	If $|S_R|$ or $|S_C|$ is too large, we stop the process and declare that the algorithm failed.  Otherwise, for $r\in\{t-1,t\}$, we let $$\widetilde{\Lambda}_r=\frac{nm}{|S_R||S_C|}\min_{\bmu^1,\ldots,\bmu^r\in \R^{|S_R|},\bmv^1,\ldots,\bmv^r\in \R^{|S_C|}}\|A|_{S_R\times S_C}-\sum_{\ell\in[r]}\bmu^\ell(\bmv^\ell)^\top\|_F^2,$$ and return $\sqrt{\widetilde{\Lambda}_t-\widetilde{\Lambda}_{t-1}}$. More precise details are given in Algorithm \ref{alg:topEigenAlg} below.
\begin{algorithm}[h!]
	\caption{Approximate the $t$-th largest Singular Value ($A$, $n$, $m$,$t$, $\eps$) }\label{alg:topEigenAlg}
	\begin{algorithmic}[1]
		\State Let $k=
		\max\Bigg\{  O\left(\frac{\max\{\log^2 n,\log^2 m\}}{\eps^2}\right), \left(\frac{1}{\eps}\right)^{O(1/\eps^2)}  \Bigg\}$
		\State Let $S_R\subseteq [n]$ such that each index $i\in[n]$ is taken to $S_R$ independently with probability $k/n$.
		\State Let $S_C\subseteq [m]$ such that each index $i\in[m]$ is taken to $S_C$ independently with probability $k/m$.
		\If {$|S_R|>2k$ or $|S_C|>2k$}
		\State \textbf{Abort}
		\EndIf
		\State Compute the following:\\ $\;\;\;\widetilde{\Lambda}_{t-1}=\frac{nm}{|S_R||S_C|}\min\limits_{\substack{\bmu^1,\ldots,\bmu^{t-1}\in \R^{|S_R|}\\\bmv^1,\ldots,\bmv^{t-1}\in \R^{|S_C|}}}\norm{A|_{S_R \times S_C}-\sum\limits_{\ell\in [t-1]}\bmu^\ell(\bmv^\ell)^\top}_F^2$
		\\$\;\;\;\widetilde{\Lambda}_{t}=\frac{nm}{|S_R||S_C|}\min_{\substack{\bmu^1,\ldots,\bmu^t\in \R^{|S_R|}\\\bmv^1,\ldots,\bmv^t\in \R^{|S_R|}}}\norm{A|_{S_R \times S_C}-\sum\limits_{\ell\in [t]}\bmu^\ell(\bmv^\ell)^\top}_F^2$
		\State \Return $z=\sqrt{\widetilde{\Lambda}_t-\widetilde{\Lambda}_{t-1}}$
		
	\end{algorithmic}
\end{algorithm}

\begin{theorem}[Restatement of Theorem~\ref{the:topEigen-intro}]\label{the:topEigen}
	Given a matrix $A \in [-L,L]^{n \times m}$, $\epsilon \in (0,1)$ let $$k=
	\max\Bigg\{  O\left(\frac{\max\{\log^2 n,\log^2 m\}}{\eps^2}\right), \left(\frac{1}{\eps}\right)^{O(1/\eps^2)}  \Bigg\}\;.$$
	Then, for every $t\le k$, Algorithm~\ref{alg:topEigenAlg} outputs a value $z$ such that with probability at least $2/3$,
	$$
	| \sigma_t(A) - z | \leq L  \sqrt{\epsilon t nm}.
	$$
\end{theorem}
In order to prove Theorem~\ref{the:topEigen}, we start by proving the following.
\begin{lemma}\label{lem:Min-Inf_general case}
	Let $A\in \R^{n\times m}$ be a matrix, and $\mathcal{A}$ be its corresponding dikernel. Then, for any $R \geq \|\mathcal{A}\|_F$, we have
	\begin{align*}
	\min_{\substack{\bmu^1,\ldots,\bmu^t\in \R^n\\\bmv^1,\ldots,\bmv^t\in \R^m}} \norm{A-\sum_{\ell\in[t]} \bmu^\ell (\bmv^\ell)^\top}_F^2
	& =
	nm\cdot \inf_{\substack{f_1,\ldots,f_t,g_1,\ldots,g_t:[0,1]\to \R: \\ \|f_\ell\|_2^2,\|g_\ell\|_2^2 \leq R\; \forall \ell \in [t]}}\norm{\mathcal{A}-\sum_{\ell\in[t]} f_\ell g_\ell^\top}_F^2\\
	& =
	nm\cdot \inf_{\substack{f_1,\ldots,f_t,g_1,\ldots,g_t:[0,1]\to \R}}\norm{\mathcal{A}-\sum_{\ell\in[t]} f_\ell g_\ell^\top}_F^2  \;.
	\end{align*}
\end{lemma}
The last minimization problem has a minimizer because the objective function is weakly continuous and coercive\footnote{A function $f:\R^n\to \R \cup \{\pm \infty\}$ is called \textit{coercive} if $\norm{\bmx}\rightarrow\infty$ implies $f(x)\rightarrow\infty$.} (See, e.g.,~\cite{Peressini:1993ug}).

\begin{proof}
	First, we show that $nm \cdot \inf\limits_{\substack{f_1,\ldots,f_t,g_1,\ldots,g_t:[0,1]\to \R: \\  \|f_\ell\|_2^2,\|g_\ell\|_2^2 \leq R\; \forall \ell \in [t] }}\norm{\mathcal{A}-\sum\limits_{\ell\in[t]} f_\ell g_\ell^\top}_F^2\le \min\limits_{\substack{\bmu^1,\ldots,\bmu^t\in \R^n\\\bmv^1,\ldots,\bmv^t\in \R^m}} \norm{A-\sum\limits_{\ell\in[t]} \bmu^\ell (\bmv^\ell)^\top}_F^2$.
	Given a solution $\bmu^1,\ldots,\bmu^t,\bmv^1,\ldots,\bmv^t$ we define functions 
	$f_1,\ldots,f_t,g_1,\ldots,g_t:[0,1]\to \R$ such that for every $\ell\in [t]$, $f_j(x)=u^\ell_{i^n(x)}$ and $g_j(x)=v^\ell_{i^m(x)}$.
	
	As we can assume $\|\bmu^\ell\|_2,\|\bmv^\ell\|_2 \leq \|A\|_F$, we have $\|f_\ell\|_2^2,\|g_\ell\|_2^2 \leq \|\mathcal{A}\|_F^2 \leq R$ for every $\ell\in[t]$.
	Then,
	\begin{align*}
	\norm{\mathcal{A}-\sum_{\ell\in[t]} f_\ell g_\ell^\top}_F^2&=\int_{0}^{1}\int_{0}^{1}\left(\mathcal{A}(x,y)-\sum_{\ell\in [t]}f_\ell(x)g_\ell(y)\right)^2 dy dx\\
	&=\int_{0}^{1}\sum_{j\in [m]}\int_{I^m_j}\left(\mathcal{A}(x,y)-\sum_{\ell\in [t]}f_\ell(x)g_\ell(y)\right)^2 dy dx\\
	&=\int_{0}^{1}\sum_{j\in [m]}\frac{1}{m}\left(\mathcal{A}(x,I^m_j)-\sum_{\ell\in [t]}f_\ell(x)v^\ell_j\right)^2 dx\\
	&=\frac{1}{nm}\sum_{i\in[n]}\sum_{j\in [m]}\left(A_{ij}-\sum_{\ell\in[t]}u^\ell_i v^\ell_j\right)^2
	=\frac{1}{nm}\norm{A-\sum_{\ell\in[t]} \bmu^\ell (\bmv^\ell)^\top}_F^2\;.
	\end{align*}
	Therefore, we have that $$nm \cdot \inf\limits_{\substack{f_1,\ldots,f_t,g_1,\ldots,g_t:[0,1]\to \R:\\  \|f_\ell\|_2^2,\|g_\ell\|_2^2  \leq R\; \forall \ell \in [t]}}\norm{\mathcal{A}-\sum\limits_{i\in[t]} f_\ell g_\ell^\top}_F^2\le \min\limits_{\substack{\bmu^1,\ldots,\bmu^t\in \R^n\\\bmv^1,\ldots,\bmv^t\in \R^m}} \norm{A-\sum\limits_{\ell\in[t]} \bmu^\ell (\bmv^\ell)^\top}_F^2$$.
	
	Next, we clearly have
	\[
	nm\cdot \inf\limits_{f_1,\ldots,f_t,g_1,\ldots,g_t:[0,1]\to \R}\norm{\mathcal{A}-\sum\limits_{\ell\in[t]} f_\ell g_\ell^\top}_F^2 \leq nm\cdot \inf\limits_{\substack{f_1,\ldots,f_t,g_1,\ldots,g_t:[0,1]\to \R: \\ \|f_\ell\|_2^2, \|g_\ell\|_2^2 \leq R\; \forall \ell \in [t]} }\norm{\mathcal{A}-\sum\limits_{\ell\in[t]} f_\ell g_\ell^\top}_F^2\;.
	\]
	
	Finally, we show that $\min\limits_{\substack{\bmu^1,\ldots,\bmu^t\in \R^n \\ \bmv^1,\ldots,\bmv^t\in \R^m}} \norm{A-\sum\limits_{\ell\in[t]} \bmu^\ell (\bmv^\ell)^\top}_F^2 \leq nm \cdot \inf\limits_{f_1,\ldots,f_t,g_1,\ldots,g_t:[0,1]\to \R}\norm{\mathcal{A}-\sum\limits_{\ell\in[t]} f_\ell g_\ell^\top}_F^2$.
	Consider the optimal solution $(f_1^*,\ldots,f_t^*,g_1^*,\ldots,g_t^*)$. %for which the objective values converge to the infimum.
	Note that we can assume that the two sets of functions $\set{f^*_\ell}$ and $\set{g^*_\ell}$ are orthogonal, since the operator $y \mapsto \sum_{\ell\in [t]}f^*_\ell(\cdot )g^*_\ell(y)$ is compact, and hence we can express it as $y \mapsto \sum_{\ell \in [t]}\sigma_\ell f'_\ell(\cdot)g'_\ell(y)$ for non-negative $\sigma_\ell\;(\ell \in [t])$ and two orthogonal sets of functions $\set{f'_\ell}$ and $\set{g'_\ell}$ by singular value decomposition of a compact operator.
	Hence, we can replace $f_\ell$ and $g_\ell$ by $\sqrt{\sigma_\ell}f'_\ell$ and $\sqrt{\sigma_\ell}g'_\ell$ without changing the objective value.
	%	by the spectral theorem of functional analysis, it can be represented as $y \mapsto \sum_{\ell\in [t]}\lambda_\ell h_\ell(\cdot)h_\ell(y)$, where $h_1,\ldots,h_t$ are orthonormal.
	%	Hence, we can replace $f_\ell$ and $g_\ell$ by $\sqrt{\lambda_\ell} h_\ell$ for each $\ell \in [t]$ without changing the objective value.	%\amargin{Is it true for functions to $\R$?}
	
	For any fixed $\ell_0\in[t]$ and $x_0\in[0,1]$, the partial derivative with respect to $f_{\ell_0}(x_0)$ is
	\begin{align*}
	\frac{\partial}{\partial f_{\ell_0} (x_0)}\norm{\mathcal{A}-\sum_{\ell\in[t]} f_\ell g_\ell^\top}_F^2 &=\frac{\partial}{\partial f_{\ell_0} (x_0)}\int_{0}^{1}\int_{0}^{1}\left(\mathcal{A}(x,y)-\sum_{\ell\in [t]}f_\ell(x)g_\ell(y)\right)^2 dy dx\\
	&=2\int_{0}^{1}g_{\ell_0}(y)\left(\sum_{\ell\in[t]}f_\ell(x_0)g_\ell(y)-\mathcal{A}(x_0,y)\right)dy\\
	&=2\sum_{\ell\in[t]}f_{\ell}(x_0)\int_{0}^{1}g_{\ell_0}(y)g_\ell(y)dy-2\int_{0}^{1}\mathcal{A}(x_0,y)g_{\ell_0}(y)dy\\
	&=2f_{\ell_0} (x_0)\normtwo{g_{\ell_0}}^2-2\sum_{i\in [m]}\int_{I^m_j}A_{i^n(x_0)j}g_{\ell_0}(y)dy.
	\end{align*}
	The partial derivatives must converge to zero almost everywhere.
	Then, by the form of the partial derivative, %for any $\epsilon > 0$, there exists $N \in \mathbb{N}$ such that, for every $\ell\in[t]$, we have $|f^N_\ell(z)-f^N_\ell(z')|<\epsilon$ for almost all $z,z'\in[0,1]$ with $i^n(z)=i^n(z')$ and $|g^N_\ell(z)-g^N_\ell(z')|<\epsilon$ for almost all $z,z'\in[0,1]$ with $i^m(z)=i^m(z')$.
	we can assume that for every $\ell\in[t]$,  $f^*_\ell$ is almost constant on each of the intervals $I^n_1,\ldots,I^n_n$ and $g^*_\ell$ is almost constant on each of the intervals $I^m_1,\ldots,I^m_m$.
	For every $\ell\in[t]$, we define $u^\ell_i=f^*_\ell(x)$ and $v^\ell_j=g^*_\ell(y)$, where $x$ and $y$ are dominant elements in $I^n_i$ and $I^m_j$, respectively.
	Thus,
	\begin{align*}
	&\norm{A-\sum_{\ell\in[t]} \bmu^\ell (\bmv^\ell)^\top}_F^2=\sum_{i\in [n]}\sum_{j\in[m]}\left(A_{ij}-\sum_{\ell\in[t]}u^\ell_i v^\ell_j\right)^2\\ %+ O(\epsilon^2 n^2 \|A\|_\infty^2 )
	&=nm\sum_{(i,j)\in[n]\times[m]}\int_{I_i}\int_{I^m_j}\left(\mathcal{A}(x,y)-\sum_{\ell\in[t]}f^*_\ell(x)g^*_\ell(y)\right)^2dx dy\\%+ O(\epsilon t nm \|A\|_{\max} )\\
	&=nm\int_{0}^{1}\int_{0}^{1}\left(\mathcal{A}(x,y)-\sum_{\ell\in[t]}f^*_\ell(x)g^*_\ell(y)\right)^2dx dy %+ O(\epsilon t nm \|A\|_{\max} ) \\
	= nm\norm{\mathcal{A}-\sum_{\ell\in[t]} f^*_\ell (g^*_\ell)^\top}_F^2 %+ O(\epsilon t nm \|A\|_{\max} )\;.
	\end{align*}
	Therefore, we have $$\min\limits_{\substack{\bmu^1,\ldots,\bmu^t\in \R^n\\\bmv^1,\ldots,\bmv^t\in \R^m}} \norm{A-\sum_{\ell\in[t]} \bmu^\ell (\bmv^\ell)^\top}_F^2 \leq nm \cdot \inf\limits_{f_1,\ldots,f_t,g_1,\ldots,g_t:[0,1]\to \R}\norm{\mathcal{A}-\sum_{\ell\in[t]} f_\ell g_\ell^\top}_F^2,$$ and the lemma follows.
\end{proof}

With this result, we can prove the theorem.

\medskip
\begin{proofof}{Theorem~\ref{the:topEigen}}
	First, we note that we can assume $L=1$ as the output of Algorithm~\ref{alg:topEigenAlg} on $A$ is $L$ times the output of Algorithm~\ref{alg:topEigenAlg} on $A/L$.
	\medskip

	Let $\mathcal{A}^2$ be a dikernel such that $\mathcal{A}^2(x,y) = \mathcal{A}(x,y)^2$ for every $x,y\in [0,1]$.
	We apply Lemma~\ref{lem:BoundTwonorm} with $\gamma=O(\eps)$ and
	\begin{align*}
	k=k_R=k_C&=
	\max\Bigg\{\frac{\max\{\log^2n,\log^2 m\}}{\gamma^2},\frac{1}{\gamma^{48}}\left(\frac{1}{\gamma^{10}}\right)^{\frac{3}{\gamma^2}}\Bigg\}\\
	&=\max\Bigg\{  O\left(\frac{\max\{\log^2 n,\log^2 m\}}{\eps^2}\right), \left(\frac{1}{\eps}\right)^{O(1/\eps^2)}  \Bigg\}\;.
	\end{align*}

	We note that since we are applying Lemma~\ref{lem:BestRankApprox} on the sampled matrix $A|_{S_R\times S_C}$, we get a bound on $t$ such that $t\le \min\{|S_R|,|S_C|\} = O(k)$.
	Then, with probability at least $2/3$, there exists a measure preserving bijection $\pi:[0,1]\to [0,1]$ such that
	\begin{align}
	\normtwo{\mathcal{A} - \pi(\mathcal{A}|_{S_R \times S_C})} \leq \epsilon,
	\quad \text{and} \quad
	\normtwo{\mathcal{A}^2 - \pi(\mathcal{A}^2|_{S_R \times S_C})} \leq \epsilon.
	\label{eq:topEigen-1}
	\end{align}
	In particular, the latter means that
	\[
	\|\mathcal{A}\|_F^2 =
	\int_0^1\int_0^1 \mathcal{A}(x,y)^2 dx dy
	=
	\int_0^1\int_0^1 \mathcal{A}|_{S_R\times S_C}(x,y)^2 dx dy \pm \epsilon
	=
	\|\mathcal{A}|_\mathcal{S_R\times S_C}\|_F^2 \pm \epsilon.
	\]
	% By Lemma~\ref{lem:Frobenius-norm-concentration}, we have
	% \begin{align}
	%   \|\mathcal{A}\| \approx \|\mathcal{A}|_{S_R\times S_C}\| \pm \epsilon L \label{eq:topEigen-2}
	% \end{align}
	% with probability at least $99/100$.
	% By the union bound, with probability at least $98/100$, we have both~\eqref{eq:topEigen-1} and~\eqref{eq:topEigen-2}.
	
	In what follows, we condition on~\eqref{eq:topEigen-1}.
	Then, we have
	\begin{align*}
	& ||\mathcal{A} - \sum_{\ell \in [t]} f_\ell g_\ell^\top||_F^2 -  ||\pi(\mathcal{A|_{S_R\times S_C}}) - \sum_{\ell \in [t]} f_\ell g_\ell^\top||_F^2\\
	& = \int_0^1\int_0^1 ({\mathcal{A}(x, y)}^2 - \pi(\mathcal{A|_{S_R\times S_C}})(x, y)^2) dx dy\\
	&\qquad - 2 \int_0^1 \int_0^1 (\mathcal{A}(x,y) - \pi(\mathcal{A|_{S_R\times S_C}})(x,y)) \sum_{\ell \in [t]} f_\ell(x) g_\ell(y)dx dy \\
	& =  \langle 1, \mathcal{A}^2 - \pi(\mathcal{A}^2|_\mathcal{S_R\times S_C}) 1\rangle -  2 \sum_{\ell \in [t]} \langle f_\ell, \mathcal{A} - \pi(\mathcal{A|_{S_R\times S_C}}) g_\ell\rangle \\
	& \leq \epsilon + 2\epsilon \sum_{\ell \in [t]} \normtwo{f_\ell} \normtwo{g_\ell}. \tag{By~\eqref{eq:topEigen-1}}
	\end{align*}
	
	Let $R = \max\set{\|\mathcal{A}\|_F^2, \|\mathcal{A|_{S_R\times S_C}}\|_F^2} = \|\mathcal{A}\|_F^2 \pm \epsilon \leq 2$.
	Then, we have
	\begin{align*}
	\Lambda_t & = \min_{\substack{\bmu^1,\ldots,\bmu^t\in \R^n\\\bmv^1,\ldots,\bmv^t\in \R^m}} \norm{A-\sum_{\ell\in[t]} \bmu^\ell (\bmv^\ell)^\top}_F^2
	= nm\cdot \inf_{\substack{f_1,\ldots,f_t,g_1,\ldots,g_t:[0,1]\to \R:\\ \|f_\ell\|_2^2,\|g_\ell\|_2^2 \leq R\;\forall \ell \in [t]} }\norm{\mathcal{A}-\sum_{\ell\in[t]} f_\ell g_\ell^\top}_F^2 \tag{By Lemma~\ref{lem:Min-Inf_general case}}\\
	& = nm\cdot \inf_{\substack{f_1,\ldots,f_t,g_1,\ldots,g_t:[0,1]\to \R: \\ \normtwo{f_\ell}^2,\normtwo{g_\ell}^2 \leq R\; \forall \ell \in [t] }}\Bigl(\norm{\pi(\mathcal{A|_{S_R\times S_C}})-\sum_{\ell\in[t]} f_\ell g_\ell^\top}_F^2 \pm (\epsilon  + 2\epsilon \sum_{\ell \in [t]}\normtwo{f_\ell} \normtwo{g_\ell})\Bigr) \\
	& = nm\cdot \inf_{\substack{f_1,\ldots,f_t,g_1,\ldots,g_t:[0,1]\to \R: \\ \|f_\ell\|_2^2,\|g_\ell\|_2^2 \leq R\; \forall \ell \in [t] }}\norm{\pi(\mathcal{A|_{S_R\times S_C}})-\sum_{\ell\in[t]} f_\ell f_\ell^\top}_F^2  \pm (\epsilon nm + 2\epsilon t R^2 nm) \\
	& = \frac{nm}{|S_R||S_C|} \cdot \min_{\substack{\bmu^1,\ldots,\bmu^t \in \R^{|S_R|}\\\bmv^1,\ldots,\bmv^t \in \R^{|S_C|}}}\norm{A|_{S_R\times S_C}-\sum_{\ell\in[t]} \bmu^\ell (\bmv^\ell)^\top}_F^2  \pm (\epsilon nm + 2\epsilon t R^2 nm) \\
	& = \tilde{\Lambda}_t \pm O(\epsilon t nm).
	\end{align*}
	In addition, from Lemma \ref{lem:BestRankApprox}, we get
	\begin{align*}
	\sigma_t(A)^2 =
	\Lambda_t - \Lambda_{t-1}
	=
	\tilde{\Lambda}_t - \tilde{\Lambda}_{t-1}
	\pm
	O(\epsilon t nm),
	\end{align*}
	which implies
	\begin{align*}
	|\sigma_t(A) - z| & \le\left|\sqrt{\tilde{\Lambda}_t - \tilde{\Lambda}_{t-1}
		\pm
		O(\epsilon t nm)} - \sqrt{\tilde{\Lambda}_t - \tilde{\Lambda}_{t-1}} \right|
	=  O(\sqrt{\epsilon t nm}).
	\end{align*}
	By replacing $\epsilon$ with $\epsilon/C$ for a sufficiently large constant $C>0$, we obtain the desired result.
\end{proofof}
\subsubsection{Experiments}\label{sec:experiments}

\begin{figure}[h!]
	\centering
	\includegraphics[width=\hsize]{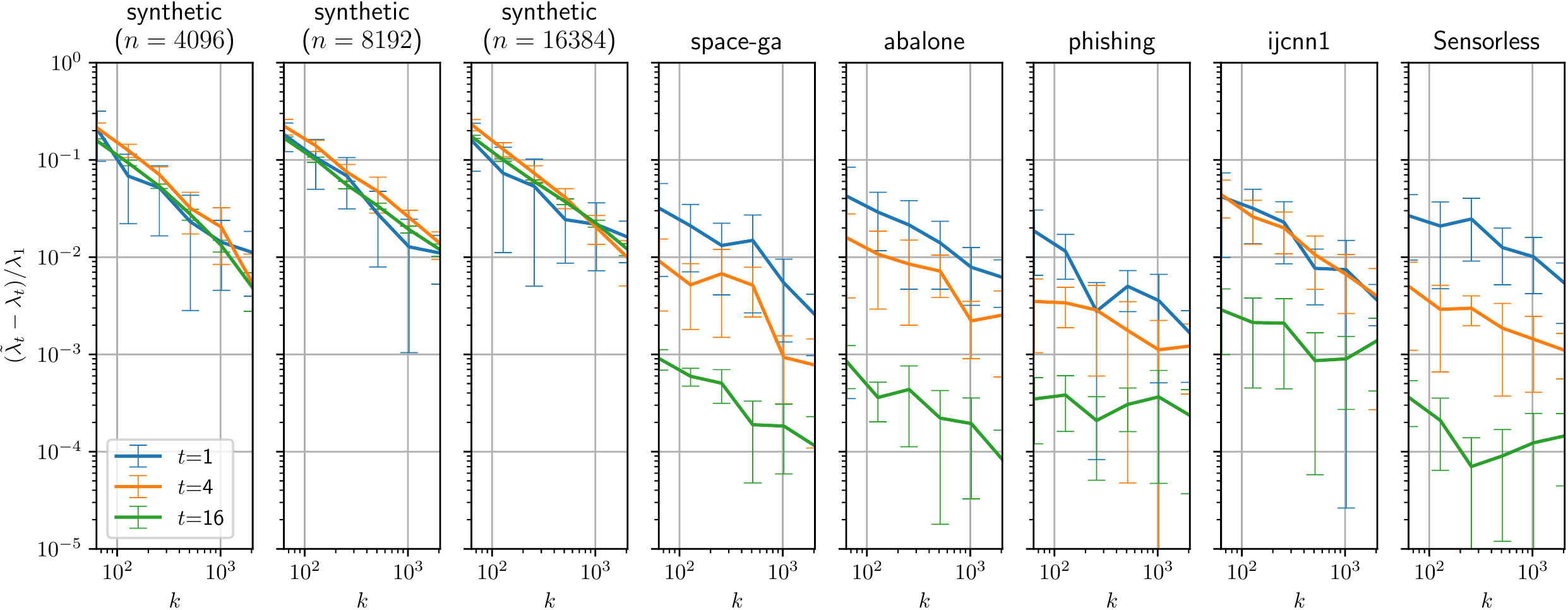}
	\caption{Accuracy. The error bar indicates the standard deviation.}\label{fig:accuracy}
\end{figure}

\begin{table}[h!]
	\centering
	\small
	\caption{Runtime in seconds for $t=16$}\label{tab:runtime}
	\begin{tabular}{|l|r||rrrrrr|r|}
		\toprule
		\multirow{2}{*}{Data} & \multirow{2}{*}{$n$} & \multicolumn{6}{c|}{$k$} & \multicolumn{1}{c|}{Power}\\
		& & 64 & 128 & 256 & 512 & 1024 & 2048 & iteration\\
		\midrule
		synthetic & 4096 & 0.03 & 0.04 & 0.07 & 0.25 & 0.90 & 3.36 & 11.70 \\
synthetic & 8192 & 0.03 & 0.04 & 0.08 & 0.24 & 0.67 & 2.34 & 48.70 \\
synthetic & 16384 & 0.02 & 0.03 & 0.06 & 0.14 & 0.54 & 1.99 & 126.10 \\
space-ga & 3107 & 0.02 & 0.02 & 0.05 & 0.20 & 0.64 & 2.63 & 5.55 \\
abalone & 4177 & 0.03 & 0.04 & 0.07 & 0.24 & 0.91 & 2.98 & 12.09 \\
phishing & 11055 & 0.02 & 0.03 & 0.05 & 0.15 & 0.59 & 2.27 & 65.38 \\
ijcnn1 & 49990 & 0.03 & 0.06 & 0.10 & 0.23 & 0.68 & 2.18 & 963.39 \\
Sensorless & 58509 & 0.03 & 0.06 & 0.11 & 0.24 & 0.67 & 2.22 & 1201.26 \\

		\bottomrule
	\end{tabular}
	
\end{table}

In this section, we experimentally demonstrate the effectiveness of our method.
We conducted experiments on a Linux server with an Intel Xeon E5-2690 (2.90 GHz) processor and 256 GB of main memory.
All the algorithms were implemented in Python. % and run using python~2.7.3.
% by conducting experiments on synthetic and real data.
%Note that the method is already shown to be effective for unconstrained quadratic optimization problems in~\cite{Hayashi:2016wh}.
Here, we consider kernel PCA, which is a representative example of constrained quadratic optimization problems.
Let $\bmx_1,\ldots,\bmx_n \in \R^d$ be data points.
For a kernel function $\textbf{ker}:\R^d \times \R^d \to \R$, we create a Gram matrix $K\in\R^{n \times n}$ in the feature space via $K_{ij}= \textbf{ker}(\bmx_i,\bmx_j)$ for each $i,j \in [n]$.
Then, we want to compute the largest few, say, $t$, eigenvalues of $K$, because it represents the maximum variance of the data points projected to a $t$-dimensional subspace in the feature space.
Note that as $K$ is positive-semidefinite, its eigenvalues are exactly its singular values, and hence we can apply our approximation algorithm for computing top singular values.

%In PCA, given a set of data points $\bmx_1,\ldots,\bmx_n \in \R^d$, we want to find a low-dimensional subspace for which the variance of the data points projected to the subspace is maximized.
%When we restrict the dimension of the subspace to be one, the variance corresponds to the largest eigenvalue of $XX^T$, where $X \in \R^{n \times d}$ is the matrix with each row corresponding to a data point.

We use synthetic and real data for our experiments.
For synthetic data, we generated a random matrix $X \in \R^{n \times 10}$ with each entry generated from a standard normal distribution.
For real data, we used space-ga ($n=3107$ and $d=6$), abalone ($n=4177$ and $d=$8), phishing ($n=11055$ and $d=$68), ijcnn1 ($n=49990$ and $d=22$), and Sensorless ($n=58509$ and $d=48$), which are provided by~\cite{Chang:2011dt}.
We adopted the radial basis function kernel $\textbf{ker}(\bmx,\bmx') = \exp(-\|\bmx-\bmx'\|_2^2/(2\sigma^2))$ with $\sigma=1$.
We implemented our method (Algorithm~\ref{alg:topEigenAlg}) using the power iteration method with $20$ iterations to compute eigenvalues of the sampled matrix, and compared it against the power iteration method with $20$ iterations on the full matrix.
We run our method 10 times for each setting.

Figure~\ref{fig:accuracy} shows the accuracy of our method.
As our method provides additive approximation, we measured the relative error with respect to $\lambda_1$, the largest eigenvalue.
Since it is computationally expensive to compute eigenvalues exactly, we regard the outputs of the power iteration method on the full matrix as the true eigenvalues.
We observe that we can achieve smaller multiplicative error as the parameter $k$ increases.
For all data, the multiplicative error against $\lambda_1$ drops to approximately 1\% by choosing $k=1024$.
%The multiplicative errors for real data are
%The approximation error becomes between 0.1\% and 1\% by choosing $k=512$.

Table~\ref{tab:runtime} shows the runtime of each method for $t=16$.
We observe that our method outperforms the power iteration method especially when $n$ is large.
This is because the runtime of our method is independent of $n$ once $k$ is determined whereas that of the power iteration method grows roughly quadratically in $n$.

\section{Acknowledgments}
The authors would like to thank Eric Blais, Kohei Hayashi, Takanori Maehara, and Hong Zhou for many useful discussions.
\bibliographystyle{alpha}
\bibliography{MinQuadApprox}

\end{document}